\documentclass[A4paper,11pt]{article}
\usepackage{amsmath,amssymb,amsfonts}

\newtheorem{theorem}{Theorem}[section]
\newtheorem{proposition}[theorem]{Proposition}
\newtheorem{corollary}[theorem]{Corollary}

\newtheorem{fact}[theorem]{Fact}
\newtheorem{definition}[theorem]{Definition}

\newtheorem{example}[theorem]{Example}

\newtheorem{question}[theorem]{Question}

\newenvironment{proof}{\mbox{\bf Proof.}}{\mbox{$\dashv$}\bigskip}

\begin{document}
\begin{center}
{\Large\bf  Computation
Environments~\footnotesize{{(1)}}}\\~\\
{\textbf{An Interactive Semantics for Turing Machines}\\
(which $\mathrm{P}$ is not equal to $\mathrm{NP}$ considering
it)~\footnote{It is the first draft of the article. The English
language is not yet perfect. This article is the first paper of
some serial related articles which would be appeared later.}}\\~
\\~
{\bf  Rasoul Ramezanian}\\
Department of Mathematical Sciences, \\
Sharif University of Technology,\\
P. O. Box 11365-9415, Tehran, Iran\\
 ramezanian@sharif.edu
\end{center}
\begin{abstract}

 \noindent To scrutinize     notions of computation and
time complexity, we introduce and formally define an interactive
model for computation that we call it the \emph{ computation
environment}. A computation environment consists of two main
parts: i) a universal processor and ii) a computist  who uses the
computability power of the universal processor to  perform
effective procedures. The notion of computation finds it meaning,
for the computist, through his \underline{interaction}   with the
universal processor.

\noindent We are interested in those computation environments
which can be considered as alternative for the real computation
environment that the human being is its computist. These
computation environments must have two properties: 1- being
physically plausible, and 2- being enough powerful.

\noindent Based on Copeland' criteria for effective procedures, we
define what a \emph{physically plausible} computation environment
is.

\noindent By being enough powerful, we mean that the universal
processor of the computation environment, must be in a way that
the computist can carry out (via the universal processor) all
logical deductions that the human being can carry out.

\noindent We construct two \emph{physically plausible} and
\emph{enough powerful} computation environments: 1- the Turing
computation environment, denoted by $E_T$, and 2-  a persistently
evolutionary computation environment, denoted  by $E_e$, which
persistently evolve in the course of executing the computations.

\noindent We prove that the equality of complexity classes
$\mathrm{P}$ and $\mathrm{NP}$ in the computation environment
$E_e$ conflicts with the \underline{free will} of the computist.

\noindent We provide an axiomatic system $\mathcal{T}$ for Turing
computability and prove that ignoring just one of the axiom of
$\mathcal{T}$, it would not be possible to derive $\mathrm{P=NP}$
from the rest of axioms.

\noindent We  prove that the computist who lives inside the
environment $E_T$, can never be confident that whether he lives in
a static environment or a persistently evolutionary one.

\end{abstract}

\section{Introduction}
\noindent We distinguish between the syntax and semantics of
Turing machines (following D. Goldin, and P. Wegner, see~5.2,
of~\cite{kn:weg07}), and  scrutinize what the computation is from
an interactive view. We  study the computation as a result of
interaction between a computist and a computer. To do this, We
propose a notion called the `computation environment', as an
interactive model for computation which regards the computation as
an interaction between the human being as a computist  and a
processor unit called the universal processor. In this way, Turing
machines (which are proposed by Turing as a model of computation)
would find their meanings based on the interaction of the human
being and the universal processor.

 A computation environment would consist of two parts: 1- a universal processor unit and
 2- a computist who uses the universal processor to perform
computation.

Based on Copeland's Criteria for effective procedures, we define
  \emph{physical plausibility} as a property  for computation
environments, and attempt to construct computation environments
which could be taken as   alternatives for the computation
environment of the real world.  That is, we are interested in
those computation environments that are physically plausible and
if the human being lives in one of them then the computability
power of the environment enables  him to carry out logical
deductions.



The famous problem $\mathrm{P}$ versus $\mathrm{NP}$ is a major
 difficult problem in computational complexity theory, and it is shown that
lots of certain techniques (such as diagonalization arguments and
relativization which are used to show that one complexity class
differs from another one) are failed to answer to this problem
(see~\cite{kn:pnp2,kn:pnp1, kn:pnp3}). However, for some
hypercomputations (see appendix~\ref{HYP}), it is proved that
$\mathrm{P\neq NP}$.

Hypercomputation extends the  capabilities of the Turing
computation via using new resources such as \emph{1)} \emph{oracle
information sources}
  \emph{2)} \emph{infinite specification}, \emph{3)} \emph{infinite computation} and
\emph{4)} the \emph{interaction} (see the appendix~\ref{intu}).
Among these four resources, the resources \emph{1,2}, and \emph{3}
do not seem physically plausible as they have infinite structures.
But the forth one, the interaction, seems physically plausible.

 It is proved for hypercomputation models which use  the resources as \emph{oracle
information}, and    \emph{infinite computation}, $\mathrm{P}$ is
not equal to $\mathrm{NP}$.

\begin{itemize}\item[] Oracle information source: an oracle machine is a Truing machine equipped with an oracle that
is capable of answering questions about the membership of a
specific set. It can be   shown that $\mathrm{P\neq NP}$ for
oracle Turing machines considering a specific oracle set (see
page~74, of~\cite{kn:arora}).

\item[] Infinite computation: it is shown that $\mathrm{P\neq NP}$
for infinite time Turing machine~\cite{kn:ITM}. Infinite time
Turing machines are a natural extension of Turing machines to
transfinite ordinal times, and the machine would be able to
operate for transfinite number of
steps~\cite{kn:ittm}.\end{itemize} In this paper, we aim to study
the famous problem $\mathrm{P}$ vs $\mathrm{NP}$ considering the
   \emph{interaction} resource.
 We introduce a new kind of
Hypercomputation called persistently evolutionary Turing machine.
A persistently evolutionary Turing machine is a Turing machine
that its inner structure  may persistently evolve during the
computation. Persistently evolutionary Turing machines can be
considered as a model of interactive computation~\cite{kn:intcom}
and should be compared with Persistent Turing
machines~independently introduced by Goldin and
Wegner~\cite{kn:weg} and Kosub~\cite{kn:kos}. We are not going to
prove that $\mathrm{P\neq NP}$ for persistently evolutionary
Turing machines similar to what is proved about infinite time
Turing machines in~\cite{kn:ITM}. Instead, we take advantage of
persistently evolutionary Turing machines and construct a
computation environment $E_e$ which its universal processor is a
persistently evolutionary Turing machine, and in this environment
$\mathrm{P}$ is not equal to $\mathrm{NP}$. In this way, we
present an interactive semantics for Turing machines and   we
prove $\mathrm{P}$ is not equal to $\mathrm{NP}$ taking this
semantics. Our work should be considered in the literature of
  interactive computation~\cite{kn:intcom} which is a new paradigm of computation.

We   construct a computation environment that we call it Turing
computation environment $E_T$.  We prove that the computationally
complexity  classes $\mathrm{P}$ and $\mathrm{NP}$ of this
environment are exactly the same complexity classes $\mathrm{P}$
and $\mathrm{NP}$ for Turing computation.  We prove that the
computist who lives inside the environment $E_T$ can never be
confident that whether he lives inside a static environment or a
persistently evolutionary one.


The paper is organized as follows:

\begin{itemize}
\item[] In Section~\ref{ETM}, we briefly review the definition of
Turing computation and time complexity, and ask some questions
which their answers  inspire  us to distinguish between syntax and
semantics of Turing machine.

 \item[] In Section~\ref{CES}, we define formally what a
computation environment is, and we introduce the Turing
computation environment that the notions of computability and time
complexity in this environment are exactly as the same as they are
for Turing machines.

\item[] In Section~\ref{peTM}, a kind of hypercomputations called
persistently evolutionary Turing machines are introduced.
Persistently evolutionary Turing machines are interactive
machines, and it is shown that they are as physically plausible as
Turing machines are.

\item[] In Section~\ref{PECE}, a persistently evolutionary
computation environment is introduced and  proved that the
equality of two  complexity classes $\mathrm{P}$ and
$\mathrm{NP}$, in this environment, conflicts with the
\underline{free will} of the computist.

\item[] In Section~\ref{AXIOM}, we provide an axiomatic system
$\mathcal{T}$ for Turing computation, and it is shown that
ignoring one of the axioms of   $\mathcal{T}$, it is not possible
to derive $\mathrm{P=NP}$ from the rest of axioms.

\end{itemize}

\section{Semantics vs. Syntax}\label{ETM}

In this section, we briefly review the standard definition  of
Turing machines and time complexity in the literature. Then we ask
and answer a question about who means and executes a Turing
machine.
\begin{definition}\label{turingd}
A  Turing machine $T$ is a tuple $T=(Q,\Sigma,\Gamma,q_0,
\delta)$, where $Q$ is a  finite set of states containing two
special states $q_0$ (initial state) and $h$ (halting state).
$\Sigma$ and $\Gamma$ are two finite sets, input and work tape
alphabets, respectively,  with $\Sigma\subseteq \Gamma$ and
$\Gamma$ has a symbol $\triangle \in \Gamma-\Sigma$, the blank
symbol. $\delta: (Q-\{h\})\times\Gamma\rightarrow
Q\times\Gamma\times\{ R,L\}$ is a partial function called the
transition function.

\begin{itemize}
\item[-] A configuration of a Turing machine $T$ is  any symbolic
form $(q,x\underline{a}y)$  where $q\in Q$, $x,y\in \Gamma^*$, and
$a\in \Gamma$.

\item[-] A configuration $C$ is a \emph{successful configuration}
whenever $C=(h,x\underline{\triangle})$ or $C=(h,
\underline{\triangle}x)$.

\item[-] A  computation path of a Turing machine $T$ on a string
$x$ is any finite sequence  $C_0C_1...C_n$ of configuration which
is started  with $C_0=(q_0, \underline{\triangle} x)$ and each
$C_{i+1}$ is obtained from $C_{i}$ by applying the transition
function $\delta$.

\item[-] We say a Turing machine $T$ accepts a string $x$, if
there is a computation path $C_0C_1...C_n$, for some $n\in
\mathbb{N}$, where  $C_0=(q_0, \underline{\triangle} x)$, and
$C_n$ is a successful configuration. We define
$L(T)\subseteq\Sigma^*$ to be the set of all strings accepted by
$T$.

\item[-] The time complexity of computing $T$ on $x$, denoted by
$time_T(x)$, is the number of configurations which appears  in the
computation path $C_0C_1...C_n$, for some $n\in \mathbb{N}$, where
$C_0=(q_0, \underline{\triangle} x)$, and $C_n$ is a successful
configuration.

\end{itemize}
\end{definition}

\begin{definition}
Let $f:\mathbb{N}\rightarrow \mathbb{N}$   and $L\subseteq
\Sigma^*$. We say the time complexity of the computation of the
language $L$ is less than $f$ whenever there exists a Turing
machine $T$ such that $L(T)=L$, and for all $x\in L$,
$time_T(x)<f(|x|)$.
\end{definition}

\begin{definition}\label{class} The computational complexity class $\mathrm{P}\subseteq 2^{\Sigma^*}$ is defined
 to be the
set of all languages that their time complexity is less than a
polynomial function. We also define the complexity class
$\mathrm{NP}\subseteq 2^{\Sigma^*}$  as follows:
\begin{itemize}
\item[] $L\in \mathrm{NP}$  iff there exists $J\in \mathrm{P}$ and
a polynomial function $q$ such that for all $x\in \Sigma^*$,
\begin{center}$x\in L\Leftrightarrow\exists y\in \Sigma^* (|y|\leq q(|x|) \wedge
(x,y)\in J)$.
\end{center}

\end{itemize}

\end{definition}

We   start by   asking and answering some questions, and inspired
by the questions, we propose the notion of computation
environment.

\begin{question}
Who \emph{executes} and \emph{means} the computation?
\end{question}

If we ask you who executes an algorithm written in a computer
programming language, say C++, you may simply answer that the CPU
(central processor unit) of the computer. Now we ask who executes
the computation in the real world?

 Consider the Turing machine

\noindent  $\ulcorner$~$T=(Q,\Sigma,\Gamma,q_0, \delta)$, where

\noindent $Q=\{q_0,q_1,h\}$, $\Sigma=\{0,1\}$,
$\Gamma=\{0,1,\triangle\}$,

\noindent $\delta=\{(q_0,\triangle)\rightarrowtail (q_1,\triangle,
R), (q_1,0)\rightarrowtail (q_1,1,R), (q_1,1)\rightarrowtail
(q_2,1,L), (q_2,0)\rightarrowtail(h,0,l)\}$~$\urcorner$.

\begin{itemize}
\item Who means  the string of alphabets written between
$\ulcorner~\urcorner$ in above?
\end{itemize}
\noindent Consider  a configuration $C=(q_1,
\triangle001\underline{1}01\triangle)$,\begin{itemize}\item who
means and executes the transition function $\delta$ on the
configuration $C$, and provides the configuration
$C'=(q_2,\triangle00\underline{1}101\triangle)$?  \end{itemize}
Turing describes an effective procedure as one that could be
performed by an infinitely patient \emph{ideal mathematician}
working with an unlimited supply of paper and
pencils~\cite{kn:ord}. Therefore, Turing knows a Turing machine as
one that is carried out by the human being

We may   answer   above questions in two different ways
  based on  what kind of entity one may assumes a computing machine
is:\begin{itemize}\item[1-] it is a mental and subjective entity
that the human being means   it, and the \textbf{brain} of the
human being as a universal processor executes it, or \item[2-] it
is a mechanical  entity that the \textbf{nature} as a universal
processor executes it and the human being is an \emph{observer} of
this execution.\end{itemize}

 In
the first view, assuming   a Turing machine $T$ and an arbitrary
string $x\in\Sigma^*$,
 it is the brain of the human being (as a universal processor) who  applies the machine $T$
on input $x$. If a Turing machine $T$ and a string $x$ are given
to the human being, then he  applies  $T$ on $x$ having an
unlimited supply of paper and pencils. He counts the number of
times that he moves his pencil to left or  right  as the time
complexity that he consumes on this computation.
 The human being does  not have access to  the brain, and ignores
 the amount of mental processes that happens inside it. Time complexity
 is defined to be the number of pencil moves not the amount of mental
 processes.

In the second view,  given   a Turing machine $M$ and an arbitrary
string $x\in\Sigma^*$,
 it is the nature (as a universal processor) who  executes the machine
 $T$
on input $x$. The human being just is a computist who knows what a
Turning machine is and sees that the nature behaves well defined,
that is, the computist cannot observe any differences between two
execution of $T$ on $x_0$ that the nature may complete in
different periods of time. Note that the nature, in order to
perform one transition of a Turing machine $T$ on a string $x$,
may need to do lots of mechanical processes that could be
invisible for the computist, and the computist ignores them and he
\emph{does not count} them as the time complexity.

 The reason that the human knows the Church-Turing
statement to be a \underline{thesis}, is that the resource which
executes effective procedures is a black for him. The human being
takes   the resource of computation (the brain or the nature)   as
a black box, and does not consider the inner structure and the
internal action of the resource, in formalizing the
  notions of computability  and time complexity.

By above arguments, we do not consider Turing machines as
\emph{\underline{autonomous}} and \emph{\underline{independent}}
entities. They are subject to a universal processor that could be
the human's brain or the nature. Nobody can build  a Turing
machine in the outer world, since a Turing machine has an infinite
tape! However, if we consider   Turing machines as   instruction
sets for a universal processor, then they find their meanings. In
this way, we differ between \emph{semantics} and \emph{syntax} of
Turing machines, following D. Goldin and P. Wegner
(5.2~of~\cite{kn:weg07}):\begin{quote}

\emph{Statements about TM \emph{(}Turing Machines\emph{)}
expressiveness, such as the Church-Turing Thesis, fundamentally
depend on their semantics. If these semantics were defined
differently, it may \emph{(}or may not\emph{)} produce an
equivalent machine. Kugel uses the terms machinery and machines to
describe the same distinction between "what a machine contains"
and "how it is used" (Kugel, 2002). He points out that if used
differently, the same machinery results in a different machine. An
example of a model that shares TM's syntax (machinery) but has
different semantics are Persistent Turing Machines; its semantics
are based on dynamic streams and
persistence.}~(page~19~of~\cite{kn:weg07})\end{quote}

There are lots of formalization   of the notion of computation,
like Church $\lambda$ calculus, Turing machines, and Markov
algorithms and etc (see, Chapter~8 of~\cite{kn:CFLA}). What all
formalization for computation have in common is that they are
purely syntactical. Mostowski says (see page~70 of
~\cite{kn:CFLA}):
\begin{quote}
\emph{However much we would like to ``mathematize" the definition
of computability, we can never get completely rid of the semantics
aspects of this concept. The process of computation is a
linguistic notion (presupposing that our notion of language is
sufficiently general); what we have to do is to delimit a class of
those functions (considered as abstract mathematical objects) for
which there exists a corresponding linguistic object (a process of
computation). }[Mostowski, p.35]
\end{quote}

Therefore, it seems reasonable that we distinguish between
semantics and syntax of a Turing machine, and propose different
semantics for it. In this way, for a Turing machine, we face with
two things:

\begin{itemize}\item[1)] the \underline{syntax}: what does a Turing machine consist of?, and
 \item[2)] the \underline{semantics}: how is a Turing machine executed?
\end{itemize} The syntax of a Turing machine is a finite set of
basic transitions coded in a finite string, and its semantics
prescribes that how a the computation starts from a configuration,
transits to other configurations and finally stops.

In this paper, we propose computation environments to present an
interactive semantics for Turing machines. We prove that regarding
this semantics $\mathrm{P}$ is not equal to $\mathrm{NP}$.

 The computation
environment would have two parts:
\begin{itemize} \item[1-] a universal processor which means and executes the syntax of computation
through interaction with a computist, and \item[2-]
 a computist  who ignores  the inner structure and the internal actions of
the universal processor in formalizing   the notion of
computability and time complexity.\end{itemize}

\begin{definition} An  instruction-set   of a processor is a finite
set of its instructions. An instruction (a command) is a single
operation of a processor that the processor interacts via it with
its outer.

\end{definition}
Note that instructions   differ  from the internal actions that a
processor may execute to perform one of its instruction.
Instructions are assumed to be atomic statements that make  it
possible for a user to communicate (interact) with the processor.

 The universal processor
(nature, or brain of the computist) executes the computation of an
instruction-set $M$ on an input string $x$.

Assuming the human being as a computist of the computation
environment of the real world, we may   say
\begin{itemize}\item[i-]Turing machines are finite instruction-sets   of the universal
processor of the real world,  and \item[ii-] if the human being
can effectively compute  a function $f$ then there exists a code
of a Turing machine as an instruction-set of the universal
processor such that the human being can communicate through the
code with the universal processor to compute the function.
\end{itemize}

 \section{Computation Environments}\label{CES}

 As it is mentioned in introduction, we   divide a  computation environment to two main parts:
 1) a  universal processor, and 2)  a computist.
  In this section, we provide a formal definitions for
 computation environments.

  In the following of the paper, we may use two terms
 `the observer' of an environment and `the god' of an environment.
 \begin{itemize}
\item An \underline{observer} of an environment is an ideal
mathematician  who does not have access to the inner structure  of
the universal processor.

\item A \underline{god} of an environment is an ideal
mathematician who has access   to the inner structure   of the
universal processor.

\end{itemize} We start some propositions and theorems   by symbol
\textbf{GV} to declare  that they are proved by the ideal
mathematician who has access to the inner structure of the
universal processor.

\begin{definition}
A \emph{universal processor} is a tuple    $U=(TBOX,SBOX, INST,
CONF)$,
 where
  \begin{itemize}
\item[1.] $INST$ is a nonempty set (assumed to be the set of all
instructions of the universal processor), and $INST_0\subseteq
INST$ is a nonempty subset called the set of  starting
instructions.

\item[2.] $CONF$ is a nonempty set called the set of
configurations such that to each $x\in \Sigma^*$,
 a unique configuration $C_{0,x}\in CONF$ is associated as the start
 configuration, and to each $C\in CONF$,  a unique string $y_C\in
 \Sigma^*$ is associated.

\end{itemize}
Two sets $INST$ and $CONF$ consist the \emph{programming language}
of the universal processor that via them a computist communicates
with the processor.
\begin{itemize}
\item[3.]$TBOX$ (the transition black box) is a total function
from $CONF\times INST$ to $CONF\cup\{\bot\}$,

\item[4.] $SBOX$ (the successful black box) is a total function
from $CONF$ to $\{YES, NO\}$,
\end{itemize} The two boxes $TBOX$ and $SBOX$ are processor units
of the universal processor.
\end{definition}

\begin{definition}~\begin{itemize}

\item[i.] A   \emph{syntax-procedure} is  a finite  set
$M\subseteq INST$ (a finite set of instructions), satisfying the
following condition (called the determination condition): for
every $C\in CONF$ either for all $\iota \in M$,
$TBOX(C,\iota)=\bot$, or at most there exists one instruction
$\tau \in M$ such that $TBOX(C,\iota)\in CONF$. We refer to the
set of all syntax-procedures by the symbol $\Xi$.

\item[ii.]We let $\upsilon: \Xi\times CONF\rightarrow
INST\cup\{\bot\}$ be a total function such that for each
syntax-procedure $M$ and $C\in CONF$, if $\upsilon(M,C)\in INST$
then $\upsilon(M,C)\in M$, and $TBOX(C,\upsilon(M,C))\in CONF$.
\end{itemize}
\end{definition}

\begin{definition} \label{CE} A \emph{computation environment} is a pair $E=(U,O)$ where $U$
is a universal processor  and $O$ is a computist   which satisfy
the following conditions:

\begin{itemize}\item[$c1.$] The computist $O$ just observes the input-output behavior of
two black boxes of the universal processor, and ignores internal
actions of the universal processor (similar to what the human
being does in formalizing notions of  computation and time
complexity)

\item[$c2.$] The computist  $O$ has a \textbf{free will} to do the
following things in  any order that he wants:

\begin{itemize} \item[$1-$] he can freely choose an arbitrary instruction $\iota \in INST$
and an arbitrary configuration $C\in CONF$ to apply the $TBOX$ on
$(C,\iota)$, and

\item[$2-$] he  can freely choose an arbitrary  configuration
$C\in CONF$ to apply the $SBOX$ on.~\footnote{ By the term
\emph{free will}, we mean there is no syntax-procedure $M_1$ that
can simulate the behavior of the computist $O$, that is, there is
no syntax-procedure $M_1$ that can recognize the ordering that the
computist $O$ inputs strings to the universal processor $U$. Or in
other words, the order that the computist  chooses strings to
input to the universal processor need not to be predetermined.}
\end{itemize}

\item[$c3.$] The computist $O$ \emph{effectively defines}
(effectively intends) a string $x\in \Sigma^*$ to be in the
\emph{language} of a syntax-procedure $M$, whenever he can
construct a sequence $C_{0}C_{1},...,C_{n}$ of configurations in
$CONF$ such that
\begin{itemize}\item $C_0=C_{0,x}$, \item each $C_i$, $i\geq 1$,
is obtained by applying $TBOX$ on $(C_{i-1},\upsilon(M,C_{i-1}))$,
\item the $SBOX$ outputs $YES$ for $C_n$, \item and either
$\upsilon(M,C_{n})=\bot$ or $TBOX(C_n, \upsilon(M,C_{n}))=\bot$.
\end{itemize} The computist $O$ calls $C_{0}C_{1},...,C_{n}$ the
successful  computation path of $M$ on $x$. The length of a
computation path is the number of configurations appeared in. We
refer to the language of a  syntax-procedure $M$, by $L(M)$. Note
that the language $L(M)$ is dependent to the computist and is  the
set of all strings, say $x$, that the computist $O$, in his living
in the computation environment, can construct a successful
computation path of $M$ for them.

\item[$c4.$]The computist $O$ \emph{effectively defines}
(effectively intends) a partial function
$f:\Sigma^*\rightarrow\Sigma^*$ through a syntax-procedure $M\in
\Xi$, whenever for $x\in \Sigma^*$,   he can construct a sequence
$C_{0}C_{1},...,C_{n}$ of configurations in $CONF$ such that
\begin{itemize}\item $C_0=C_{0,x}$, \item each $C_i$, $i\geq 1$,
is obtained by applying $TBOX$ on $(C_{i-1},\upsilon(M,C_{i-1}))$,
\item the $SBOX$ outputs $YES$ for $C_n$, \item and either
$\upsilon(M,C_{n})=\bot$ or $TBOX(C_n, \upsilon(M,C_{n}))=\bot$,

\item $y_{C_n}=f(x)$.
\end{itemize}

\item[$c5.$] The semantics of a syntax-procedure  $M\in \Xi$ (the
way it is executed in the environment) is determined through the
interaction of the computist with the universal processor.

\item[$c6.$] The computist may start to apply the universal
processor $U$ on a  syntax-procedure  $M$ and a string $x$,
however he does not have to keep on the computation until the
successful box outputs $Yes$. The computist can leave the
computation $M$ on $x$ at any stage of his activity  and choose
freely any other syntax-procedure $M'$ and any other string $x'$
to apply $U$ on them.

\item[$c7.$] The computist $O$ defines the \textbf{time
complexity} of computing a syntax-procedure $M$ on an input string
$x$, denoted by $time_M(x)$,   to be $n$, for some $n\in
\mathbb{N}$, whenever he constructs  a successful computation path
of the syntax-procedure $M$ on $x$ with length $n$.

\item[$c8.$] The inner structure of the black boxes of the
universal processor may evolve through their inner processes on
inputs, however they have to behave well-defined through time, and
the computist cannot realize the evolution, because he does not
have access to the inner part of the universal processor.

\item[$c9.$] Let $f:\mathbb{N}\rightarrow \mathbb{N}$  and
$L\subseteq \Sigma^*$. The computist says that the time complexity
of the computation of the language $L$ in the computation
environment $E$, is less than $f$ whenever there exists a
syntax-procedure  $M\in \Xi$ such that the language defined by the
computist via $M$, i.e., $L(M)$, is equal to $L$, and for all
$x\in L$, $time_M(x)<f(|x|)$.

\item[$c10.$] The computist  defines the time complexity class
$\mathrm{P}_E\subseteq 2^{\Sigma^*}$ to be the set of all
languages that he can effectively define  in polynomial time. He
also defines the complexity class $\mathrm{NP}_E\subseteq
2^{\Sigma^*}$ as follows:
\begin{itemize}
\item[] $L\in \mathrm{NP}_E$  iff there exists $J\in \mathrm{P}_E$
and  a polynomial function $q$ such that for all $x\in
\Sigma^*$,\begin{center} $x\in L\Leftrightarrow\exists y\in
\Sigma^* (|y|\leq q(|x|) \wedge (x,y)\in J)$.\end{center}
\end{itemize}
\item[$c11.$] \textbf{\emph{Logical Deduction}}. We say a
computation environment is enough powerful whenever the computist
can carry out logical deduction using the universal processor.


\end{itemize}
\end{definition}

By definition of computation environments, we attempt to
mathematically model the situation that the human being has with
the notion of computation in the real world. In a computation
environment,
\begin{quote} a language $L$ (a function $f$) is effectively
calculable
 if and only if there exists a syntax-procedure $M\in\Xi$ such that  the computist  can define
the language $L$ (the function $f$) through it.\end{quote} In this
way, the syntax-procedures find  their meaning via the interaction
of the computist with the universal processor of the environment.

\begin{example}\label{tce}{\em \textbf{(The Turing Computation Environment)}.

Let $Q_T=\{h\}\cup\{q_i\mid   i\in \mathbb{N}\cup\{0\}\}$,
$\Sigma,\Gamma$ be two finite set  with $\Sigma\subseteq \Gamma$
and $\Gamma$ has a symbol $\triangle \in \Gamma-\Sigma$.

The universal processor $U_s=(TBOX_s,SBOX_s, INST_s, CONF_s)$ is
defined as follows:
\begin{itemize}
\item[1)] $INST_s=\{[(q,a)\rightarrow(p,b,D)]\mid p,q\in Q_T,
a,b\in \Gamma, D\in\{R,L\}\}$,

$(INST_s)_0=\{[(q,a)\rightarrow(p,b,D)]\in INST_s\mid q=q_0 \}$,
and


\item[2)] $CONF_s=\{(q,x\underline{a}z)\mid q\in Q_T, x,z\in
\Gamma^*, a\in \Gamma\}$,   for each $x\in \Sigma^*$,
$C_{0,x}=(q_0,\underline{\triangle} x)$, and for each
$C=(q,x\underline{a}z)\in CONF_s$, $y_C=xaz$.

\end{itemize}
Note that the programming language of $U_s$ is exactly the
standard syntax of configurations and transition functions of
Turing machines.
\begin{itemize}

\item[3)] Let $C=(q,xb_1\underline{a}b_2y)$ be an arbitrary
configuration then

\begin{itemize} \item $TBOX_s(C,[(q,a)\rightarrow (p,c,R)])$ is defined to be $C'=(p,xb_1c\underline{b_2}y)$,

\item $TBOX_s(C,[(q,a)\rightarrow (p,c,L)])$ is defined to be
$C'=(p,x\underline{b_1}cb_2y)$, and

\item for other cases $TBOX_s$ is defined to be $\bot$.

\end{itemize}
\item[4-] Let $C\in CONF_s$ be arbitrary
\begin{itemize}
\item if $C=(h,\underline{\triangle}x)$ then $SBOX_s(C)$ is
defined to be $YES$,

\item if $C=(h,x\underline{\triangle})$ then $SBOX_s(C)$ is
defined to be $YES$, and

\item otherwise $SBOX_s(C)$ is defined to be $NO$.
\end{itemize}

\item[5-] For each $M\in \Xi_s$, and $C=(q,x\underline{a}y)\in
CONF_s$,   if there exists $[(q,a)\rightarrow(p,b,D)]\in M$ for
some $p\in Q_T, b\in \Gamma$, and $D\in \{R,L\}$, then
$\upsilon(M,C)$ is defined to be $[(q,a)\rightarrow(p,b,D)]$ else
it is defined to be  $\bot$.
\end{itemize}
We call the computation environment $E_T=(U_s,O_f)$ (where $O_f$
is a computist with the free will) the \underline{Turing
Computation Environment}. }
\end{example}
Note that each syntax-procedure $M\in \Xi_s$ of the Turing
computation environment $E_T$ is exactly a transition function of
a Turing machine. Let $TM=\{T=(Q, \Sigma, \Gamma,q_0,\delta)\mid
Q\subseteq Q_T\}$ be the set of all Turing machines (according to
definition~\ref{turingd}) that the set of their states is a finite
subset of $Q_T$. Define two functions $\mathcal{F}_1:TM\rightarrow
\Xi_s$ and $\mathcal{F}_2:\Xi_s\rightarrow TM$ as follows:

\begin{itemize}
\item for each $T\in TM$, $\mathcal{F}_1(T)=\delta$,

\item for each $M\in \Xi_s$, $\mathcal{F}_2(M)=(Q, \Sigma,
\Gamma,q_0,M)$, where $Q$ is the set of all states in $Q_T$ which
appeared in the instructions of $M$.
\end{itemize}

One may easily prove that, for all $T\in TM$, if the human being
computes $L(T)$ in the real world, and the computist $O_f$
computes $L(\mathcal{F}_1(T))$ in the Turing computation
environment $E_T$ then $L(T)=L(\mathcal{F}_1(T))$.

Also, for all $M\in \Xi_s$, if the computist $O_f$ computes $L(M)$
in the Turing computation environment $E_T$, and the human being
computes $\mathcal{F}_2(M)$ in the real world then
$L(M)=L(\mathcal{F}_2(M))$.

\begin{theorem}\label{lang} (\textbf{GV}).

\begin{itemize}
\item[1-]For each syntax-procedure $M\in \Xi_s$, there exists a
Turing machine $T=\mathcal{F}_2(M)$ such that for every $x\in
\Sigma^*$, $x\in L(M)$ and $time_M(x)=n$ iff $x\in L(T)$ and
$time_T(x)=n$.

\item[2-] For each Turing machine $T$, there exists a
syntax-procedure $M\in \Xi_s$, $M=\mathcal{F}_1(T)$, such that for
every $x\in \Sigma^*$, $x\in L(T)$ and $time_T(x)=n$ iff $x\in
L(M)$ and $time_M(x)=n$.

\end{itemize}
\end{theorem}\begin{proof} It is straightforward.\end{proof}

The above theorem is true from the god's view. The computist who
lives inside  the environment cannot prove the above theorem as he
does not have access to the inner structure of the universal
processor $U_s$. However he is always free to propose a hypothesis
about its computation environment  like Church and Turing did.
Note that since we (the human beings) do not have access to the
inner structure of our mind, we call the Church Turing statement
to be a thesis. If we  know that how effective procedures are
executed by our mind, then we should not call the Church Turing
statement to be a thesis.

\begin{corollary} (\textbf{GV}).~ The time complexity classes in the Turing computation environment are exactly the
same classes of Turing machines. Particulary,

\begin{itemize}
\item[1-] $\mathrm{P}_{E_T}=\mathrm{P}$, and

\item[2-] $\mathrm{NP}_{E_T}=\mathrm{NP}$.
\end{itemize}

\end{corollary}

\begin{definition}
A universal processor $U=(TBOX,SBOX, INST, CONF)$ is called to be
static, whenever the inner structure of two boxes $TBOX$ and
$SBOX$ does not change due to interaction with the computist. It
is called persistently evolutionary whenever the the inner
structure of at least one boxes changes but persistently, i.e., in
the way that the boxes work well-defined.
\end{definition} The universal processor of the Turing computation environment defined in the
example~\ref{tce} is static, in Section~\ref{PECE}, We introduce a
computation  environment which universal processor persistently
evolve.

 We may construct lots of computation environments. But we
are interested in those computation environments that are

\begin{itemize}\item
\emph{physically plausible}   \item \emph{enough
powerful,}\end{itemize} and thus could be considered as
alternatives of the computation environment of the real world.

 By enough powerful,
we mean that the computability power of  the universal processor
must be enough high such that the the computist who lives in the
environment can carry out all deductions in first order logic that
the ideal mathematician can carry out in the real world.

Note that if a theory $\mathcal{T}$ in the first order logic is
recursively enumerable then the set of all statements that could
be logically derived from the theory is also recursively
enumerable, that is, deduction is Turing computable. Hence, We
consider the following definition.

\begin{definition}\label{enpo}
A computation environment $E=(U,O)$ is enough powerful, whenever
for every recursively enumerable language $L$, there exists a
syntax-procedure $M$ in the environment that the computist can
define $L$ through $M$, i.e., $L(M)=L$.
\end{definition}

\begin{proposition} The Turing computation environment $E_T$ is
enough powerful.
\end{proposition}\begin{proof} It is straightforward.
\end{proof}

By physical plausibility, we mean the inner structures of the
black boxes have to be as physically plausible as a Turing machine
is.  Copeland gives four criteria that any set of instructions
that makes up a procedure should satisfy in order to be
characterized as an \emph{effective procedure} (see the
Church-Turing Thesis, Chapter 2, Page~21,~\cite{kn:AS}):
\begin{itemize}
\item[1)] Each instruction is expressed by a finite string of
Alphabet, that is the elements of $INST$ must be definable  by
finite strings.

\item[2)] The instructions are basic and produce the result in one
step (or at most in  a finite number of steps).

\item[3)] They can be carried out by the computist (the human
being) unaided by any machinery, ignoring the internal actions of
the universal processor (his mind).

\item[4)] They demand no insight on the part of the computist
carrying it out (the computist just simply inputs an instruction
and a configuration in the transition box, and receives the next
configuration. No insight is needed).

\end{itemize}

Also a procedure has to be   \emph{finite}, \emph{definite} (each
instruction must be clear and unambiguous), and \emph{effective}
(every instruction (like $[(p,a)\rightarrowtail (q,b,R)$]) has to
be sufficiently basic)~\footnote{See the definition of algorithm
presented in~page~21,~\cite{kn:AS}  which is accepted by most
computer scientist and engineers}.

Inspired by Copeland's criteria, we define physical plausibility
for computation environments.

\begin{definition}\label{pypl}
A computation  environment $E=(U,O)$, where
$U=(TBOX,SBOX,INST,CONF)$ is \emph{physically plausible} whenever
\begin{itemize}
\item[a1.] The inner structure of two black boxes $TBOX$ and
$SBOX$ are as physically plausible as the human being knows the
Turing machines are.

\item[a2.] The instructions in $INST$ can be described in finite
words,

\item[a3.] The configurations in $CONF$ can be described in finite
words.
\end{itemize}
\end{definition}

\begin{corollary} Every computation environment $E=(U=(TBOX,SBOX,INST,CONF),O)$
which \begin{itemize}\item[1.] $INST=INST_s$, \item[2.]
$CONF=CONF_s$ (see the definition of $E_T$ in example~\ref{tce}),
\item[3.] and the inner structure of $TBOX$ and $SBOX$ are
physically plausible machines,\end{itemize} is physically
plausible.
\end{corollary}\begin{proof} Since the instructions in $INST_s$
and configurations in $CONF_s$ are describable in finite strings,
the items a2, and a3 of definition~\ref{pypl} is fulfilled.
\end{proof}

\begin{question}
Is there a \emph{physically plausible} and \emph{enough powerful}
computation environment $E$ such that $\mathrm{P}_E$ is not equal
to $\mathrm{NP}_E$   for the computist who uses the universal
processor freely?
\end{question}
In the next section, we introduce persistently evolutionary Turing
machines as a kind of hypercomputations. Then in  a sequel, based
on this hypercomputation, we construct a physically plausible
computation environment that the successful box of its universal
processor is a persistently evolutionary Turing machine, and
$\mathrm{P}$ is not equal to $\mathrm{NP}$ in the environment.

\section{Persistently Evolutionary Turing Machines} \label{peTM}

 The notion of being constructive in Brouwer's
intuitionism  goes beyond the Church-Turing thesis (see the
appendix~\ref{intu}). In this section, inspired by Brouwer's
choice sequences, we introduce persistently evolutionary Turing
machines. A persistently evolutionary Turing machine is a machine
that its inner structure during its computation on any input may
evolve~\footnote{Similar computing machines (dynamic machines)
were presented in the author's PhD thesis~\cite{kn:phdr}.}. But
this evolution is in the way that if a computist  does not have
access to the inner structure of the machine  then he cannot
recognize whether the machine evolves or not.

\begin{definition}
Let $M_1$ and $M_2$ be  two (deterministic) Turing machines, and
$x\in \Sigma^*$ be arbitrary. We say $M_1$ and $M_2$ are
$x$-equivalent, denoted by $M_1\equiv_x M_2$ whenever

\begin{itemize}

\item[] if one of the two machines $M_1$ and $M_2$ outputs $y$ for
input  $x$, then the other one also outputs the same $y$ for the
same input $x$.
\end{itemize}
\end{definition}

\begin{definition} A Persistently evolutionary Turing machine is a
couple $N=(\langle z_0,z_1,...,z_i\rangle, f)$ where $\langle
z_0,z_1,...,z_i\rangle$ is a \emph{growing sequence}  of codes of
deterministic Turing machines, and   $f$ (called the persistently
evolutionary function) is a computable partial function from
$\Sigma^*\times \Sigma^*$ to $\Sigma^*$ such that for any code of
a Turing machine, say $y$, and any string $x\in \Sigma^*$, if $y$
halts on $x$ then $f(y,x)$ is defined and it is a code of a new
Turing machine. The function $f$ satisfies the following condition
(that we call it the persistent condition):
\begin{itemize}
\item[-] for every finite sequence  $(x_1,x_2,...,x_n, x_{n+1})$
of $\Sigma^*$, and every Turing machine $y_0$, if
$y_1=f(y_0,x_1)$, $y_2=f(y_1,x_2)$,..., $y_n=f(y_{n-1},x_n)$ and
$y_{n+1}=f(y_n,x_{n+1})$ are defined then we have for all $0\leq
i\leq n$, $y_{n+1}\equiv_{x_{i+1}} y_i$.
\end{itemize}
Whenever an input $x$ is given to $N$, the output of the
evolutionary machine $N$ is computed according to the $z_i$
($z_1=f(z_0,x_1)$, $z_2=f(z_1,x_2)$,... $z_i=f(z_{i-1},x_i)$,
where $x_1,x_2,...,x_i$ are strings that sequentially  have given
as inputs to $N$  \underline{until now}, and $N$ has halted  for
all of them), and the machine evolves to $N=(\langle
z_0,z_1,...,z_i,z_{i+1}= f(z_i,x)\rangle, f)$ (if $z_i$ halts for
$x$). Note that since the evolution happens persistently, as soon
as $N$ provides an output for an input $x$, if we will input the
same $x$ to $N$ again in future, then $N$ provides the same output
as before. It says that, the machine $N$ behaves well-defined as
an input-output black box.
\end{definition}

The persistently evolutionary Turing computation could be assumed
as one of forms of Hypercomputation~\cite{kn:TO, kn:AS}. Note that
at each moment of time, a persistently evolutionary Turing machine
has a finite structure, but during computations on inputs, its
structure may change. So it is not possible to encode a
persistently evolutionary Turing machine in a finite word. A
persistently evolutionary Turing machine is an interactive machine
(see Chapter~5 of~\cite{kn:AS}, and~\cite{kn:intcom}) that evolves
according to how it interacts with its users. One may compare
persistently evolutionary Turing machines with  persistent Turing
machines (PTM). The   difference between these two kinds of
machines is that PTM's are static machines that  transform  input
streams (infinite sequence of strings) to output streams
persistently~\cite{kn:PTM,kn:weg,kn:kos}, whereas persistently
evolutionary Turing machines are evolutionary machines that
transform input strings to output strings.

 We
may start with two equal persistently evolutionary Turing machines
$N_1=(z_0,f)$ and $N_2=(z'_0,f')$ where $z_0=z'_0$ and $f=f'$, but
as we input strings to two machines in different orders the
machines may evolve in different ways.

\begin{example}
Let $z_0$ be a code of an arbitrary Turing machine, and
$I:\Sigma^*\rightarrow\Sigma^*$ be the identity function, $I(x)=x$
for all $x\in \Sigma^*$. Then $N=(z_0,I)$ is a persistently
evolutionary Turing machine. Thus any Turing machine can be
considered as a persistently evolutionary Turing machine.
\end{example}

\begin{fact} \label{fact} A user who just has access to the
input-output behavior of an evolutionary Turing machine $N$,
cannot become conscious whether the machine $N$ evolves or not. In
other words, if we put a persistently evolutionary machine in a
black box, a computist  can never be aware that whether it is a
(static) Turing machine   in the black box, or   a persistently
evolutionary one.
\end{fact}

\begin{corollary}
Let $E=(U,O)$ be a computation environment that the computist $O$
does not have access to the inner structure of the processor $U$.
The computist $O$ can never be aware that whether its environment
persistently evolves or not.
\end{corollary}

In the next example, we let $\mathrm{NFA}_1$ be the class of all
nondeterministic finite automata that for each $M\in
\mathrm{NFA}_1$, each  state $q$ of $M$, and $a\in \Sigma$, there
exists at most one transition from $q$ with label $a$.

\begin{example}\label{auto}   We define a function $h:
\mathrm{NFA}_1\times\Sigma^*\rightarrow \mathrm{NFA}_1$ as
follows. Let $M\in\mathrm{NFA}_1$, $M=\langle Q,
q_0,\Sigma=\{0,1\},\delta:Q\times\Sigma\rightarrow Q, F\subseteq
Q\rangle$,  and $x\in \Sigma^*$. Suppose $x=a_0a_1\cdots a_k$
where $a_i\in \Sigma$. Applying the automata $M$ on $x$, one of
the three following cases may happen:
\begin{itemize}
\item[1.] The automata $M$ could read all $a_0,a_1\cdots ,a_k$
successfully and stops in an accepting state.    Then we let
$h(M,x)=M$.

\item[2.] The automata $M$ could read all $a_0,a_1\cdots ,a_k$
successfully and stops in a state $p$ which is not an accepting
state. If the automata $M$ can  transit from $p$ to an  accepting
state by reading \emph{one } alphabet, then we define $h(M,x)=M$.
If it cannot transit (to an accepting state)  then we define
$h(M,x)$ to be a new automata $M'=\langle Q,
q'_0,\Sigma=\{0,1\},\delta':Q'\times\Sigma\rightarrow Q',
F'\subseteq Q'\rangle$, where $Q'=Q$, $\delta'=\delta$,
$F'=F\cup\{p\}$.

\item[3.] The automata $M$ cannot read all $a_0,a_1\cdots ,a_k$
successfully,and after reading a part of $x$, say $a_0a_1\cdots
a_i$, $0\leq i\leq k$, it crashes in  a state $q$ that
$\delta(q,a_{i+1})$ is not defined. In this case, we let $h(M,x)$
be a new automata $M'=\langle Q,
q'_0,\Sigma=\{0,1\},\delta':Q'\times\Sigma\rightarrow Q',
F'\subseteq Q'\rangle$, where $Q'= Q\cup \{s_{i+1},s_{i+2},\cdots,
s_k\}$ (all $s_{i+1},s_{i+2},\cdots, s_k$ are different states
that does not belong to $Q$), $\delta'=\delta\cup
\{(q,a_{i+1},s_{i+1}), (s_{i+1},a_{i+2},s_{i+2}),\cdots,
(s_{k-1},a_k,s_{k})\}$, and $F'=F\cup\{ s_k\}$.
\end{itemize}
For each $M\in NFA_1$, we  let $T_M$ be a Turing machine that for
each input $x\in \Sigma^*$, the machine $T_M$ first constructs the
automata $h(M,x)$, and if $h(M,x)$ accepts $x$, then $T_M$ outputs
$1$, else it outputs $0$.

We define the persistently evolutionary Turing machine
$PT_1=\langle \lfloor T_{M_0}\rfloor, f\rangle$, where
$M^0=\langle Q^0=\{q_0\}, q_0,\Sigma=\{0,1\},\delta^0=\emptyset,
F^0=\emptyset\rangle$, and $f(\lfloor T_M\rfloor, x)=\lfloor
T_{h(M,x)}\rfloor$.
\end{example}

A language which is intended by a persistently evolutionary Turing
machine  is an unfinished object (similar to Brouwer's choice
sequences). At each stage of time, only a finite part of it,  is
recognized with the human being who intends the language.
Persistently evolutionary Turing machines exist in time and are
temporal dynamic mental
constructions~(see~page~16~of~\cite{kn:BH}). Also a language
intended by a persistently evolutionary Turing machine is
\emph{user-dependent}. Two users      with two persistently
evolutionary Turing  machines $N_{1}$ and $N_{2}$ with the same
initial structure may intend two different languages. A language
which is defined through a persistently evolutionary Turing
machine is not predetermined and is dependent to the free will of
the user.

 The
initial structure $N=(z_0,f)$ of a persistently evolutionary
Turing machine is constructed at a particular moment of time, and
then evolves as the user chooses further strings to input. For
persistently evolutionary Turing machines  what remains invariant
is the character of the machine as a evolutionary machine, an
evolution that started at a particular point in time and preserves
well-definedness. The machine evolves but it is the same machine
that evolves and the machine is an individual unfinished object.
Note that the user is not allowed to reset the machine and goes
back to past. It is because that the evolution is a characteristic
of persistently evolutionary Turing machines.

\begin{definition}~
\begin{itemize}\item[]
We say a persistently evolutionary Turing machine $N=(\langle
z_0,z_1,...,z_n\rangle,f)$ evolves in a feasible time whenever the
function $f$ can be computed in polynomial time.

\item[] We say a persistently evolutionary Turing machine
$N=(\langle z_0,z_1,...,z_n\rangle,f)$ works in polynomial time,
whenever there exists a polynomial function $p$, such that for
every $x\in \Sigma^*$, if we input $x$ to $N$, the machine $z_i$
provides the outputs in less than $p(|x|)$ times, and the function
$f$ constructs the function $z_{i+1}=f(z_i,x)$ in less than
$p(|x|+|z_i|)$.\end{itemize}
\end{definition}

\begin{proposition}\label{feas}
The persistent evolutionary Turing machine $PT_1$ works in
polynomial time.
\end{proposition}\begin{proof}
It is straightforward.
\end{proof}

\subsection{Physical Plausibility}~\label{ppl}
 Hypercomputation extends the  capabilities of   Turing
computation via using new resources such as 1- infinite memory, 2-
infinite specification, 3- infinite computation and 4- the
interaction (see the appendix~\ref{intu}). Among these four
resources, the resources 1,2, and 3 do not seem physically
plausible as they have infinite structures.  But the forth one,
the interaction, seems physically plausible for the human being.
The human being interacts with its environment and   it could be
possible that its environment persistently evolve because of
interaction. The persistently evolutionary Turing machines use two
resources to be   Hypercomputations 1- evolution, and 2-
interaction.  Both of these two resources could be accepted by the
human being as physically plausible resources  (Biological
structures evolve in Darwin theory). One may hesitate to accept
that Persistently evolutionary Turing machines are physically
plausible due to his presupposition that the real world is a
Turing machine, but  if he releases himself  from this confinement
then it seems to him that persistently evolutionary Turing
machines are as physically plausible as Turing machines are. One
can implement a persistently evolutionary Turing machine on his
personal computer (assuming that the computer has an infinite
memory. Note that the same assumption is needed for executing
Turing machines on a computer). The difference between Turing
machines and persistently evolutionary Turing machines is that the
languages that the first category recognize are predetermined,
whereas in the second category, we can recognize an unfinished and
not predetermined language that depends on our interactions.

Therefore, we may list the following items for why    persistently
evolutionary Turing machines are physically plausible:

\begin{itemize}
\item[1-] The persistently evolutionary Turing computation is a
kind of interactive computation~\cite{kn:intcom} which today is
accepted as a new paradigm of computation by some computer
scientist.

\item[2-] The persistently  evolutionary Turing machines are as
plausible as Turing machines are. Both machines can be simulated
by a personal computer (assuming that the computer has an infinite
memory).

\item[3-] In Brouwer's intuitionism, choice sequences are accepted
as constructive mathematical objects. Choice sequences are
growing, unfinished objects. Therefore, the persistently
evolutionary Turing machines as growing unfinished objects are
acceptable as constructive mathematical objects in Brouwer's
intuitionism.
\end{itemize}

\subsection{Persistent  Evolution}
Suppose   $\mathbf{B}$ to be an input-output black box. For the
observer who does not  have access to the inner structure of the
black box, it is not possible   to get aware that whether the
inner structure of the black box persistently evolves or not. We
only sense a change whenever we discover that an event which has
been  sensed before is not going to be sensed similar to past.
Persistently evolution always respects the past. As soon as, an
agent experiences an event then whenever in future he examines the
same event, he will experience it similar to past. However,
persistent  evolution effects the future which has not been
predetermined, and not experienced by the agent yet. Therefore,

\begin{quote}it is not possible for an agent to distinguish between
persistent  evolution and being static based on the history of his
observation.
\end{quote}

\subsection{Free Will}\label{FWL}
Suppose $N=(z_0,f)$ is a persistently evolutionary Turing machine.
The machine $N$ could evolve in different ways due to the free
will of the user who chooses freely strings to input to $N$.  For
example, consider the persistently evolutionary Turing machine
$PT_1$. If you input two strings $111$, and $11$ respectively, the
machine $PT_1$ outputs $1$ for the first  input and outputs $0$
for the second one. Inputting the string $111$ made the machine to
evolve such that it cannot accept $11$ anymore. The time has sink
back to past, and the machine $PT_1$ evolved.

However, you \emph{had} the free will to input first $11$ and then
$111$, and if this case had happened, the machine would have
accepted both of them!

The user of a persistently evolutionary Turing machine cannot
change the past by his free will, but he can effect the future by
his free choices. As soon as a persistently evolutionary Turing
machine evolves, it has been evolved, and it is not possible to go
back to past. When a persistently evolutionary machine evolves, it
is the same machine that evolves, and the evolution is a part of
the entity of the machine.

However, since the user has the free will, he can effect the
future. The future is not necessary predetermined, and the user
can make lots of different futures due to his free will. For
example, let $L$ be the set of all strings that the evolutionary
machine $PT_1$, during the interaction with a user (with me),
outputs $1$ for them.  The language $L$ is not predetermined, and
it is a growing and an unfinished object (similar to choice
sequences~\ref{intu}). Consider the formula
\begin{center}$\phi:=(\exists k\in\mathbb{N})(\forall n>k)(\exists x\in
\Sigma^*)(|x|=n\wedge x\in L)$.\end{center} At each stage of time,
having evidence for the truth of formula $\phi$ conflicts with the
free will of the user, and we never could have evidence for truth
of $\phi$. It is because at each stage of time, only for a finite
number of strings in $\Sigma^*$, it is predetermined that whether
they are in $L$ or not. Let $m\in \mathbb{N}$ be such that $m$
would be grater than the length  of all strings that until now are
determined to be in $L$. The user via his free will can input all
strings with length $m+1$ to the machine $PT_1$ respectively. The
machine $PT_1$ outputs $1$ for all of them, and evolves such that
$L\cap\{x\in \Sigma^*\mid |x|=m\}$ would be empty! Also, at each
stage of time having evidence for  $\neg\phi$   conflicts with the
free will of the computist. Again, let $m\in \mathbb{N}$ be such
that $m$ would be grater than the length  of all strings that
until now are determined to be in $L$. The user via his free will
can input a string  with length $m$ to the machine $PT_1$. The
machine $PT_1$ outputs $1$ for it.

\begin{theorem}\label{pt1}
Let $L\subseteq\Sigma^*$ be the growing language intended by a
user through the persistently evolutionary Turing machine $PT_1$.
let \begin{center}$\phi:=(\exists k\in\mathbb{N})(\forall
n>k)(\exists x\in \Sigma^*)(|x|=n\wedge x\in L)$.\end{center}  We
 can never have evidence not for $\varphi$ and not for $\neg\varphi$,
if we assume that the user has the free will.
\end{theorem}\begin{proof} See the above argument.
\end{proof}
\section{A Persistently Evolutionary Computation
Environment}\label{PECE}

As we discussed  in Section~\ref{ETM}, we do not assume  Turing
machines as autonomous and independent entities. The meaning of a
Turing machine is subject to the interaction between the computist
(the human being), and the universal processor (human's brain, or
the nature). Therefore, if the universal processor changes then
the meanings of Turing machines change (although their syntax
remain unchanged).
 We
may start an argument by asking a question
\begin{question}
What about if the human's brain or the nature persistently evolve
in the course that a computation   is executed?
\end{question}
 In this section, we propose
a computation environment, denoted by $E_e=(U_e,O_f)$,  which its
universal processor $U_e$ persistently evolves, and its computist
$O_f$ has the free will. The computation environment
$E_e=(U_e,O_f)$ is such that $\mathrm{P}_{E_e}$ is not equal to
$\mathrm{NP}_{E_e}$ in it, for the computist who has the free
will.

\begin{definition}\label{Ee} The universal processor $U_e=(TBOX_e,SBOX_e, INST_e,
CONF_e)$ is defined as follows. \begin{itemize}

\item[] Two sets $INST_e$ and $CONF_e$ are defined to be the same
$INST_s$ and $CONF_s$ in the Turing computation environment~(see
example~\ref{tce}) respectively, and consequently the set of
procedures of the universal processor $U_e$, i.e., $\Xi_e$ is the
same $\Xi_s$.

\item[] The transition box $TBOX_e$ is also defined similar to the
transition box $TBOX_s$ of the Turing computation environment
$E_T$~(see example~\ref{tce}).

\item[] The successful box $SBOX_e$ is defined as follows: let
$C\in CONF_e$ be arbitrary
\begin{itemize}
\item if $C=(h,\underline{\triangle}x)$ then $SBOX_s(C)=YES$,

\item if $C=(h,x\underline{\triangle})$ then the $SBOX_e$ works
exactly similar to the way that the persistently evolutionary
Turing machine $PT_1$ (introduced in example~\ref{auto}) performs
on input $x$, and if $PT_1$ outputs $1$, the successful box
outputs $YES$, and

\item otherwise $SBOX_e(C)=NO$.
\end{itemize}
\end{itemize}
\end{definition}
  Note that the successful box of the universal processor $U_e$ is
a persistently evolutionary Turing machine. For the computist
$O_f$ the set of syntax-procedures (algorithms) in the environment
$E_e$ is the same set of syntax-procedures in the environment
$E_T$. However the semantics of a syntax-procedure in the
environment $E_e$ may be different from its semantics in the
environment $E_T$.

In the computation environment $E_e=(U_e,O_f)$, the
syntax-procedures in $\Xi_e$  satisfy the four criteria that
Copeland gives that any set of instruction that makes up an
algorithm should satisfy to be characterized as an
\underline{\emph{effective procedure}}  (see the Church-Turing
Thesis, Chapter 2, page~21,~\cite{kn:AS}). The only thing that one
should pay attention  is that:   the meaning of procedures may
persistently change in the environment for the computist $O_f$.
But it does not mean that they are unambiguous for the computist.
The computist does not have access to the universal processor, and
is not aware of the persistent evolution. The computist using a
pencil and paper, writes a configuration on the paper, then using
the universal processor (that could be its brain, and the
computist ignores its processes) writes the computation path  on
the paper, and what happens inside the universal processor is
ignored by the computist. Also note that

\begin{quote} the universal processor $U_e$ is a physically plausible
engine~(see~\ref{ppl}), and it could be assumed as an alternative
for the universal processor of the real world.\end{quote}

\begin{theorem}\label{prf} (\textbf{GV}). Every recursively  enumerable language can be recognized in the environment
$E_e$. That is, for every recursively enumerable language $L$,
there exists $M\in \Xi_e$ such that $L=L(M)$.
\end{theorem}\begin{proof}
 It is straightforward. For each Turing machine $T$, one may
 construct a syntax-procedure  $M\in\Xi_e$ such that $L(T)=L(M)$, and it is not possible for
 the computist to obtain any configurations
 $(h,y\underline{\triangle})$, $y\in\Sigma^*$, using the procedure $M$. Then the
 universal processor $U_e$ behaves exactly similar to the
 universal processor $U_s$ for the procedure $M$.
\end{proof}

\begin{corollary} (\textbf{GV}). The computation environment $E_e=(U_e,O_f)$ is \emph{enough
powerful} (see definition~\ref{enpo}).
\end{corollary}

\begin{proposition}\label{pee} (\textbf{GV}). The complexity class $\mathrm{P}$ is a subset of
$\mathrm{P}_{E_e}$.
\end{proposition}\begin{proof}
It is straightforward.\end{proof}

The converse of   Theorems~\ref{prf}, and~\ref{pee} do not hold
true. The computist can effectively define some languages in
polynomial time which are unfinished and not predetermined (see
the proof of theorem~\ref{subex}).

\begin{definition}
Let $\phi$ be a statement about the syntax-procedures  in $\Xi_e$.
We say the truth of $\phi$ is consistent with the \underline{free
will} of the computist, whenever independent of any way that the
computist interacts with the universal processor, the statement
$\phi$ always holds true in the environment.
\end{definition}

Any formula  that forces the future to be predetermined could
conflict  with the free will of the computist (also
see~\ref{FWL}). In following, we prove that
$\mathrm{P}_{E_e}=\mathrm{NP}_{E_e}$ conflicts with the free will
of the computist.

 The item c2 of
definition~\ref{CE}, considerers the free will as a characteristic
of any computist who lives in  a computation environment,
therefore, if a formula conflicts with the free will, then from
the god's view, the formula would not be valid for the computation
environment.

\begin{definition}
We say a function $f:\mathbb{N}\rightarrow \mathbb{N}$ is
sub-exponential, whenever there exists $t\in \mathbb{N}$ such that
for all $n>t$, $f(n)<2^n$.
\end{definition}

\begin{theorem}\label{subex} \textbf{(GV)}. There exists a syntax-procedure $M\in \Xi_e$ such that

\begin{itemize}
\item the language $L(M)$ that the computist $O_f$ effectively
defines (effectively intends) through $M$ belongs to the class
$P_{E_e}$, and

\item there exists no syntax-procedure $M'\in \Xi_e$, such that
the language that the computist defines through $M'$, i.e.,
$L(M')$, is equal to $L'= \{x\in \Sigma^*\mid\exists y
(|y|=|x|\wedge y\in L(M))\}$, and for some $k\in \mathbb{N}$, for
all $x\in L(M')$, if $|x|>k$ then
\begin{center}
$time_M(x)\leq f(|x|)$
\end{center}
where $f:\mathbb{N}\rightarrow \mathrm{N}$ is a sub-exponential
function.

\end{itemize}
\end{theorem}\begin{proof}
 Consider the following syntax-procedure $M\in \Xi_e$
\begin{itemize} \item[] $\Sigma=\{0,1\},\Gamma=\{0,1,\triangle\},$

\item[]$M=\{[(q_0,\triangle)\rightarrow(h,\triangle,R)],
[(h,0)\rightarrow(h,0,R)],[(h,1)\rightarrow(h,1,R)]\}$.
\end{itemize} The intension of the word $M$ is persistently
evolutionary in the computation environment $(U_e,O_f)$ for the
computist $O_f$. That is, as the computist  chooses a string $x\in
\Sigma^*$ to check whether $x$ is an element of   $L(M)$, the
\emph{intension of $M$} persistently changes in the environment,
though since  the computist  does not have access to the inner
structure of the universal processor $U_e$, he never gets aware of
changes.
\begin{quote}
The language $L(M)$ is intended \emph{deterministically} through
the syntax-procedure  $M$ by the computist  $O_f$, but it is not a
\emph{predetermined}  language.
\end{quote}

 It is obvious that the language that the computist  $O_f$ intends by the
syntax-procedure  $M$, i.e., $L(M)$ belongs to $\mathrm{P}_{E_e}$
(due to definition of time complexity in definition~\ref{CE}).

 Let $L'=\{x\in \Sigma^*\mid\exists y
(|y|=|x|\wedge y\in L(M))\}$. It is again obvious that $L'$
belongs to $\mathrm{NP}_{E_e}$.

Suppose there  exists  a syntax-procedure $M'\in \Xi_e$ that the
computist can intends $L'$ by $M'$ in time complexity less than a
sub-exponential function $f$. Then for some $k\in \mathbb{N}$, for
all $x$ with length greater than $k$, $x$ belongs to $L'$
whenever\begin{itemize} \item[] the computist constructs  a
successful computation path $C_{0,x}C_{1,x},...,C_{n,x}$ of the
syntax-procedure $M'$ on $x$, for some $n\leq
f(|x|)$.\end{itemize}

 Let $m_1\in \mathbb{N}$ be the maximum length of those
strings $y\in \Sigma^*$ that   \underline{until now} are accepted
by the persistently evolutionary Turing machine $PT_1$  available
in the successful box of $U_e$. Define $m=\max(m_1,k)$.

For every $y\in \Sigma^*$, let
$path(y):=C_{0,y}C_{1,y},...,C_{f(|y|),y}$ be the computational
path of syntax-procedure $M'$ on the string $y$. The $path(y)$ can
be generated by the transition box of $U_e$. Note that $TBOX_e$ is
a static machine and does not evolve, thus $path(y)$ is
independent of the behavior of the computist. Let
\begin{center}
$S(w)=\{C_{j,y}\mid C_{j,y}\in path(y) \wedge \exists x\in
\Sigma^*(C_{i,y}=(h,x\underline{\triangle} ))\}$
\end{center}and
\begin{center}
$H(y)=\{x\in \Sigma^*\mid \exists C_{j,y}\in path(y)
(C_{j,y}=(h,x\underline{\triangle} ))\}$
\end{center}
We refer by $|H(y)|$ to the number of elements of $H(y)$, we have
$|H(y)|\leq f(|y|)$ if $|y|>k$. Also let $E(y)=H(y)\cap \{x\in
\Sigma^*\mid |x|=|y|\}$, and $D(y)=H(y)\cap \{x\in \Sigma^*\mid
|x|=|y|+2\}$.

Consider $w\in \Sigma^*$ with $|w|>m$. Two cases are possible:
either $S(w)=\emptyset$ or it is inhabited.

\underline{Consider the first case}. $S(w)=\emptyset$.

The computist chooses a string $w\in \Sigma^*$ with length greater
than $m$ and applies the universal processor $U_e$ on the
syntax-procedure $M'$ and the string $w$, and recognizes whether
$w$ is in $L(M')$ or not. since the set $S(w)$ is empty, the
execution of $M'$ on $w$ does not make the $SBOX_e$ to evolve, and
it remains unchanged.

If the computist realizes that $w\in L(M')$ then it means that
there exists a string $v\in \Sigma^*$ such that $|v|=|w|$ and
$v\in L(M)$. But length $v$ is greater than $m$, and it
contradicts with free will of the computist  that he can apply the
universal processor in any arbitrary orders on instructions  and
configurations (see c2 of definition~\ref{CE}). Since the
computist  has the free will, he can apply the universal processor
$U_e$ on the syntax-procedure $M$ and all strings in $\Sigma^*$
with length $|v|+1$ sequentially. As the length of $v$ is greater
than $m$, all strings with length $|v|+1$ are accepted by the
persistently evolutionary Turing machine $PT_1$ (see item-3 of
example~\ref{auto}) available in $SBOX_e$,  and the universal
processor evolves. Now if the computist  applies the universal
processor  on the syntax-procedure $M$ and the string $v$, then
$v$ cannot be an element of    $L(M)$ (see the item-2 of
example~\ref{auto}) and $v\not\in L(M)$, contradiction.

If the computist  realizes that $w\not\in L(M')$ then it means
that for all strings $v\in \Sigma^*$, $|v|=|w|$, we have $v\not\in
L(M)$. But it contradicts with  the free will of the computist
again. As the length $w$ is greater than $m$, the computist  may
choose a string $z$ with $|z|=|w|$ and by the item-3 of
example~\ref{auto}, we have $z\in L(M)$, contradiction.

\underline{Consider the second case}. $S(w)\neq\emptyset$.

 Suppose that before computing $M'$ on $w$,  the computist applies the universal processor on the
syntax-procedure  $M$ and all strings $v0$'s, for  $v\in E(w)$,
and then applies the universal processor on the syntax-procedure
$M$ and strings $v0$'s, for $v\in D(w)$ respectively. Since
$|w|>m$, the computist would have $v0\in L(M)$ for all $v\in
E(w)\cup D(w)$, and the successful box of the  universal processor
$U_e$ evolves through computing $M$ on $v0$'s.

 After that, the computist  applies the universal
processor on $M'$ and $w$. The universal processor already evolved
in the way that the successful box outputs $No$ for all
configuration in $\{C_{i,w}\in S(w)\mid \exists x\in E(w)\cup
D(w)(C_{i,w}=(h,x\underline{\triangle}))\}$. Either the computist
finds $w\in L(M')$ or $w\not\in L(M')$.
 Suppose the first case
happens and $w\in L(M')$. It contradicts with the free will of the
computist. The computist can apply the universal processor on $M$
and strings $v0$, $|v|=|w|$ sequentially, and would make $\{v0\in
\Sigma^*\mid |v|=|w|\}\subseteq L(M)$. Then the successful box
evolves in the way that, it will output $No$ for all
configurations $(h,v\underline{\triangle})$, $|v|=|w|$, and thus
there would exist no $v\in L(M)\cap\{x\in \Sigma^*\mid |x|=|w|\}$.

Suppose the second case happens and $w\not\in L(M')$. Since
$|H(w)|<f(|w|)< 2^{|w|}$, during the computation of $M'$ on $w$,
only $f(|w|)$ numbers of configurations  of the form
$(h,x\underline{\triangle})$, $x\in\{v0\mid |v|=|w|\}\cup \{v1\mid
|v|=|w|\}$ are given as input to the successful box. Therefore
there exists a string $z\in \{x\in \Sigma^*\mid |x|=|w|\}$ such
that none of its successors have been input to the persistently
evolutionary Turing machine $PT_1$, and if the computist chooses
$z$ and computes $M$ on it, then $z\in L(M)$. Contradiction.
\end{proof}

Note that the proof of Theorem~\ref{subex} cannot be carried by
the computist $O_f$ who lives in the environment. The above proof
is done by the god, the ideal mathematician who has   access to
the inner structure of the universal processor.
\begin{theorem}\label{freew} \textbf{(GV)}.~{\footnotesize{\emph{Free Will}}}~\footnote{Readers may see~\cite{kn:parad},
where the notion of  the \underline{free will} is used to settle
the well-known Surprise Exam Paradox. The notion of the
\underline{free will} is also discussed in author's PhD
thesis~\cite{kn:phdr}, and in~\cite{kn:deci}.}.\begin{center}
$\mathrm{NP}_{E_e}\neq\mathrm{P}_{E_e}$.\end{center}
\end{theorem}\begin{proof} It is a consequence of
theorem~\ref{subex}.
\end{proof}

The above theorem says that there exists a syntax-procedure  $M\in
\Xi_e$ that the computist  $O_f$ can effectively define   the
language $L(M)$ in polynomial time through the syntax-procedure
$M$, such that for $L'=\{x\in \Sigma^*\mid\exists y (|y|=|x|\wedge
y\in L(M))\}$, if $L'$ can be effectively defined through a
syntax-procedure $M'\in \Xi_e$ then it forces the computist  to
use the universal processor in some certain orders, which
conflicts with the free will of the computist~\footnote{Note that
the above argument says that if  the language of a
syntax-procedure $M$ is not predetermined and the computist  who
computes the language through $M$ has free will, then future is
always undetermined and any assumption that forces the future to
be determined conflicts with the free will. The statement
$\mathrm{P=NP}$ is one of those assumptions that forces the future
behavior of the computist to be determined.}.

\begin{corollary} The computation environment $E_e=(U_e,O_f)$
is a mathematical model for the Copeland's criteria of effective
procedures. Thus, \begin{center} one cannot derive
$\mathrm{P}=\mathrm{NP}$ from the criteria introduced by Copeland
for effective procedures.\end{center}
\end{corollary}

\begin{corollary} There exists a \emph{physically plausible} and \emph{enough powerful} computation
environment which $\mathrm{P}$ is not equal to $\mathrm{NP}$ in.
\end{corollary}\begin{proof}
Persistently evolutionary Turing machines seem  as physically
plausible as Turing machines are. The structure of the transition
box of the universal processor  $U_e$ is a Turing machine and the
structure of its successful box is a Persistently evolutionary
Turing machine. Therefore $E_e=(U_e,O_f)$ is a physically
plausible computation environment. Also $E_e$ is enough powerful
by Theorem~\ref{prf}.
\end{proof}

\subsection{Complexity Equivalence}
In this part, we define an  equivalency between universal
processors from  time complexity view. We prove that the universal
processor $U_s$ of Turing computation environment $E_T$ is
polynomial time equivalent to the universal processor $U_e$ of the
persistently evolutionary computation environment $E_e$.

\begin{definition} Let $U_1=(TBOX_1,SBOX_1,INST_s,CONF_s)$, and

$U_2=(TBOX_2,SBOX_2,INST_s,CONF_s)$ be two universal processors
where their instructions and configurations are the same
instructions and configurations of the Turing computation
environment $E_T$. Also let $f:\mathbb{N}\rightarrow \mathbb{N}$
be a function. We say two universal processor $U_1$ and $U_2$ are
\emph{$f$-complexity equivalent} whenever:

\begin{itemize}
\item[1-] for each $C=(q,x\underline{a}y)\in CONF_s$, and each
$\iota\in INST_s$, if $TBOX_1(C,\iota)=C'$ for some $C'\in
CONF_s$, then there exist $\tau_1,\tau_2,...,\tau_k\in INST_s$,
for some $k\leq f(|xay|)$, such that for some $C_1,C_2,...,C_k\in
CONF_s$, $TBOX_2(C,\iota_1)=C_1$,   for all $2\leq i\leq k$,
$TBOX_2(C_{i-1},\iota_{i-1})=C_i$ and $C_k=C'$.

\item[2-] for each $C=(q,x\underline{a}y)\in CONF_s$, and each
$\iota\in INST_s$, if $TBOX_2(C,\iota)=C'$ for some $C'\in
CONF_s$, then there exist $\tau_1,\tau_2,...,\tau_k\in INST_s$,
 for some $k\leq f(|xay|)$, such that  for some
$C_1,C_2,...,C_k\in CONF_s$, $TBOX_1(C,\iota_1)=C_1$,   for all
$2\leq i\leq k$, $TBOX_1(C_{i-1},\iota_{i-1})=C_i$ and $C_k=C'$.

\item[3-] for each $C=(q,x\underline{a}y)\in CONF_s$, if
$SBOX_1(C)=YES$, then there exist $\tau_1,\tau_2,...,\tau_k\in
INST_s$, for some $k\leq f(|xay|)$ such that for some
$C_1,C_2,...,C_k\in CONF_s$, $TBOX_2(C,\iota_1)=C_1$,   for all
$2\leq i\leq k$, $TBOX_2(C_{i-1},\iota_{i-1})=C_i$ and
$SBOX_2(C_k)=YES$.

\item[4-] for each $C=(q,x\underline{a}y)\in CONF_s$, if
$SBOX_2(C)=YES$, then there exist $\tau_1,\tau_2,...,\tau_k\in
INST_s$,  for some $k\leq f(|xay|)$, such that for some
$C_1,C_2,...,C_k\in CONF_s$, $TBOX_1(C,\iota_1)=C_1$,   for all
$2\leq i\leq k$, $TBOX_1(C_{i-1},\iota_{i-1})=C_i$ and
$SBOX_1(C_k)=YES$.
\end{itemize}
\end{definition}

Suppose you, as a computist, have two personal computer on your
desk. The CPU of one of them is the universal processor $U_1$, and
the CPU of the other one is the universal processor $U_2$, the
programming language that you can communicate with the processor
are the same, and you know that

\begin{itemize}
\item[-] if you can transit from a configuration
$C=(q,x\underline{a}y)$ and an instruction $\iota$ to a
configuration $C'$ using the transition box $TBOX_1$ ($TBOX_2$)
only one time, then you can transit from the same configuration
$C$  to the same configuration $C'$ using the transition box
$TBOX_2$ ($TBOX_1$) at most $f(|xay|)$ times, and

\item[-] if you input    a configuration $C=(q,x\underline{a}y)$
to the $SBOX_1$ ($SBOX_2$),  and receives $YES$, then you can
transit from the configuration $C$ to a configuration $C'$ using
the transition box $TBOX_2$ ($TBOX_1$)  at most $f(|xay|)$ times,
such that the successful box  $SBOX_2$ ($SBOX_1$) outputs $YES$
for $C'$.
\end{itemize} Then, for you as a computist, it seems (wrongly) that if you
can do some work   using the processor $U_1$ in $m$ times, then
you may do the same work   using the   processor $U_2$ for at most
$2m\times f(m)$ times, and vice versa.

\begin{proposition}
Two universal processors $U_s$ and $U_e$ are $f$-complexity
equivalent where $f(n)=2n$ for all $n\in \mathbb{N}$.
\end{proposition}
\begin{proof} It is straightforward.
\end{proof}


\section{An Axiomatic System for Turing Computability}\label{AXIOM}
In this section, we provide an  axiomatic system for Turing
computation.  Let $CONF=CONF_s$, and  $INST=INST_s$ (where
$CONF_s, INST_s$ are defined in example~\ref{tce}. Also let $\Xi$
be the set of all finite subsets of $INST$ such that for every
$M\in \Xi$, for every two instructions
$\iota_1=[(q_1,a_1)\rightarrow (p_1,b_1,D_1)]$,
$\iota_2=[(q_2,a_2)\rightarrow (p_2,b_2,D_2)]$, if
$\iota_1,\iota_2\in M$, $q_1=q_2$, and $a_1=a_2$ then $p_1=p_2$,
$b_1=b_2$, and $D_1=D_2$.  Consider
\begin{itemize}
\item[] a function symbol $TB: CONF\times INST\rightarrow CONF
\cup\{\perp\}$, and

\item[] a predicate symbol $SB: CONF\rightarrow \{YES,NO\}$.
\end{itemize}
For each $x\in\Sigma^*$, define
$C_{0,x}=(q_0,\underline{\triangle} x)$.

\begin{definition} For each $x\in \Sigma^*$, for each $M\in \Xi$,
we define

\begin{itemize} \item[] $x\in L(M):\equiv$ there exists $n\in
\mathbb{N}$, there exists $\tau_1,\tau_2,...,\tau_n\in M$, there
exists $C_1,C_2,...,C_n\in CONF$, such that
$C_1=TB(C_{0,x},\tau_1)$, and $C_i=TB(C_{i-1}, \tau_i)$ for
$1<i\leq n$, and $SB(C_n)=YES$, and for all $\iota\in M$,
$TB(C_n,\iota)=\perp$.

\item[] $time_M(x)=n$ whenever there exists
$\tau_1,\tau_2,...,\tau_n\in M$, there exists $C_1,C_2,...,C_n\in
CONF$, such that $C_1=TB(C_{0,x},\tau_1)$, and $C_i=TB(C_{i-1},
\tau_i)$ for $1<i\leq n$, and $SB(C_n)=YES$, and for all $\iota\in
M$, $TB(C_n,\iota)=\perp$.

\end{itemize}

\end{definition}
An axiomatic system $\mathcal{T}$ for Turing computation consists
of following axioms A1-A4.

\begin{itemize}
\item[\textbf{A1.}] for all $x,y\in \Sigma^*$, for all
$a,b_1,b_2,c\in \Sigma$, for all $q,p\in Q$, \begin{center}
$TB((q,xb_1\underline{a}b_2y),[(q,a)\rightarrow
(p,c,R)])=(p,xb_1c\underline{b_2}y)$, and
\end{center}
 \begin{center}
$TB((q,xb_1\underline{a}b_2y),[(q,a)\rightarrow
(p,c,L)])=(p,x\underline{b_1}cb_2y)$.
\end{center}

\item[\textbf{A2.}] If $SB(C)=YES$ then for some $x\in \Sigma^*$,
either $C=(h,\underline{\triangle}x)$ or
$C=(h,x\underline{\triangle})$.

\item[\textbf{A3.}] For all $x\in \Sigma^*$, if
$C=(h,\underline{\triangle}x)$ then $SB(C)=YES$.

\item[\textbf{A4.}] For all $x\in \Sigma^*$, if
$C=(h,x\underline{\triangle})$ then $SB(C)=YES$.
\end{itemize}
Let $TM=\{T=(Q, \Sigma, \Gamma,q_0,\delta)\mid Q\subseteq Q_T\}$
be the set of all Turing machines (see definition~\ref{turingd}).
Define two functions $\mathcal{F}_1:TM\rightarrow \Xi$ and
$\mathcal{F}_2:\Xi\rightarrow TM$ as follows:

\begin{itemize}
\item for each Turing machine $T\in TM$,
$\mathcal{F}_1(T)=\delta$,

\item for each $M\in \Xi$, $\mathcal{F}_2(M)=(Q, \Sigma,
\Gamma,q_0,M)$, where $Q$ is the set of all states appeared in the
instructions of $M$.
\end{itemize}

\begin{theorem} For each $x\in \Sigma^*$, for every $T\in TM$, for
every $M\in \Xi$,
\begin{itemize}
\item $x\in L(T)$ iff $\mathcal{T}\vdash x\in
L(\mathcal{F}_1(T))$.

\item $time_T(x)=n$ iff $\mathcal{T}\vdash
time_{\mathcal{F}_1(T)}(x)=n$.

\item $\mathcal{T}\vdash x\in L(M)$ iff $x\in
L(\mathcal{F}_2(M))$.

\item $\mathcal{T}\vdash time_M(x)=n$ iff
$time_{\mathcal{F}_2(M)}(x)=n$.
\end{itemize}
\end{theorem}\begin{proof} It is straightforward.
\end{proof}

The above theorem declares that the axiomatic system $\mathcal{T}$
is an appropriate axiomatization of Turing computation.

\begin{theorem}
Let $\mathcal{T}'$ be the axiomatic system $\mathcal{T}-\{A_4\}$.
\begin{center}
$\mathcal{T}'\not\vdash \mathrm{P=NP}$.
\end{center}
\end{theorem}\begin{proof} We need to provide a model for the theory
$\mathcal{T}'$ such that $\mathrm{P}$ is not equal to
$\mathrm{NP}$ in the model. Interpret the predicate symbol $SB$ to
be the successful box $SBOX_e$ of computation environment $E_e$,
and the function symbol $TB$ to be the transition box $TBOX_e$.
All axioms in $\mathcal{T}'$ are satisfied by this interpretation.
On the other hand, by Theorem~\ref{freew}, in this model, we have
$\mathrm{P\neq NP}$.~\footnote{Consider axiom
\begin{itemize}\item[$\mathbf{A'2.}$]   $SB(C)=YES$ if and only if  for some
$x\in \Sigma^*$ $C=(h,\underline{\triangle}x)$.
\end{itemize} One may check that $\Omega=\{\mathbf{A1},\mathbf{A'2}\}$ is a theory for
Turing computability, and for every $M\in \Xi$, for every $x\in
\Sigma^*$, $\Gamma'\vdash x\in L(M)$ iff $\Omega\vdash x\in L(M)$,
and for every $M\in \Xi$, for every $x\in \Sigma^*$, for every
$n\in \mathbb{N}$, $\Gamma'\vdash time_M(x)=n$ iff $\Omega\vdash
time_M(x)=n$.}
\end{proof}

\section{$\mathrm{P}$ vs $\mathrm{NP}$ for the Computist}

 We claim that the computist who lives inside the Turing
computation environment $E_T$ can  never be confident  that
whether he lives inside a static environment or a persistently
evolutionary one. It is simply because of the three following
facts:
\begin{itemize}
\item[1-] The language (syntax) of two computation environments
$E_T$ and $E_e$ are the same. That is, $INST_e=INST_s$,
$CONF_e=CONF_s$, and $\Xi_e=\Xi_s$.

\item[2-] The universal processor of the Turing computation
environment is a black box for the computist. The computist can
never know that whether the universal processor of its environment
persistently evolves or not. It is because an observer who doe not
have access to the inner structure of a black box, can never
distinguish that whether the black box persistently evolves or
not.

\item[3-] If the universal processor of an environment
persistently evolves, then it is possible that there are
syntax-procedures in $\Xi$ which their languages are not
predetermined and $\mathrm{P=NP}$ conflicts with the free will of
the computist. For example, $\mathrm{P}$ is not equal to
$\mathrm{NP}$ in the computation environment $E_e$.
\end{itemize}

In follows, we show that for the computist  who lives in one of
the environments $E_T$ or $E_e$, it is not possible to get aware
that actually he lives in a static   environment or a persistently
evolutionary one.

\begin{definition}\textbf{ (\emph{BLACK BOX})}. Let $X$ and $Y$ be two
sets,

\begin{itemize}
\item  an \emph{input-output} black box $\mathbf{B}$, for an
observer $O$, is a box that \begin{itemize} \item  The observer
$O$ does not see the inner instruction of the box, and

\item the observer $O$ chooses elements in $X$, and input them to
the box, and receives elements in $Y$ as output.
\end{itemize}\item   We say an input-output black box $\mathbf{B}$ behaves
\underline{well-defined} for an observer $O$, whenever if the
observer $O$ inputs $x_0$ to the black box, and the black box
outputs $y_0$ at a stage of time, then whenever in future if the
observer $O$ inputs the same $x_0$ again, the black box outputs
the same $y_0$.

\item We say a well-defined black box is static (or is not
order-sensitive) whenever for all $n\in\mathbb{ N}$, for every
$x_1,x_2,...,x_n\in X$, for every permutation $\sigma$ on
$\{1,2,...,n\}$, if the  observer   inputs $x_1,x_2,...,x_n$
respectively to $B$ once, and receives
$y_1=\mathbf{B}x_1,y_2=\mathbf{B}x_2,...,y_n=\mathbf{B}x_n$, and
then the god (the one who has access to the inner structure of the
black box) resets the inner structure of the black box
$\mathbf{B}$, and after reset, the observer inputs
$x_{\sigma(1)},x_{\sigma(2)},...,x_{\sigma(n)}$ respectively to
$\mathbf{B}$, then the outputs of $\mathbf{B}$ for each $x_i$
would be the same already output $y_i$.

\item An observer who does not have access to the inner structure
of a  well-defined black box $\mathbf{B}$, and cannot reset the
inner structure  black box $\mathbf{B}$, can never get aware
whether the black box is static or order-sensitive.
\end{itemize}
\end{definition}
\begin{example} Let $X=Y=\mathbb{N}$, and $\mathbf{B}$ be the black box that works as follows.
 I hide in the black box, and do the following strategy. I
plan to output $1$ for each $n$ before the observer gives the
black box $5$ or $13$ as inputs. If the observer  gives $5$ as
input (and has not yet given $13$ before) then   after that time,
I output $2$ for all future inputs that have not been already
given to the black box. For those natural numbers that  the
observer have already given them as input, I still output the same
$1$.
 If the observer  gives $13$ as input (and has not yet given
$5$ before) then   after that time, I output $3$ for all future
inputs that have not been already given to the black box. For
those natural numbers that the observer has already given them as
input, I still output the same $1$. It is easy to verify that

\begin{itemize}
\item[1-] The black box $\mathbf{B}$ works well-defined.

\item[2-] The black box is order-sensitive.

\item[3-] For the observer $O$, it is not possible to be aware
that the black box is order sensitive. He always could assume that
actually  a machine Turing is in the black box.

\item[4-] The function that the observer computes via the black
box $B$ is not a predetermined function, but the observer can
never find out that whether the function is predetermined or not!
\end{itemize}

\end{example}

\begin{proposition}~
\begin{itemize}
\item For every finite sets of pairs $S=\{(i_k,o_k)\mid 1\leq
k\leq n, n\in \mathbb{N}, i_k,o_k\in \Sigma^*\}$, there exists a
Turing machine $T$ such that for all $(i_k,o_k)\in S$, if we give
$i_k$ as an input to $T$, the Turing machine $T$ outputs $o_k$.

\item For every finite sets of pairs $S=\{(i_k,o_k)\mid 1\leq
k\leq n, n\in \mathbb{N}, i_k,o_k\in \Sigma^*\}$, there exists a
Persistent evolutionary Turing machine $N$ such that for all
$(i_k,o_k)\in S$, if we give $i_k$ as an input to $N$, the
Persistent evolutionary  machine $N$ outputs $o_k$.

\end{itemize}
\end{proposition}\begin{proof}
It is straightforward.
\end{proof}

\begin{corollary}
Let $B$ be a an input-output black box for an observer. At each
stage of time, the observer has observed only a finite set of
input-output pairs. By the pervious proposition, at each stage of
time, the observer knows both the following cases to be possible:
\begin{itemize}
\item[1-] There exists a Turing machine inside the black box $B$.

\item[2-] There exists a Persistent evolutionary Turing machine
inside the black Box $B$.
\end{itemize}
\end{corollary}

\begin{fact}
In any computation environment $E=(U,O)$, since the computist $O$
cannot reset the universal processor and goes back to past, he can
\underline{never} realize that actually he lives in a static
environment or in  a persistently evolutionary one.
\end{fact}

Therefore, the computist who lives in the Turing computation
environment can never have evidence from his observation that the
universal processor which means to syntax-procedures in $\Xi_s$ is
static or  persistently evolutionary.

The following argument can be carried out in Epistemic
logic~\cite{kn:dit3}.

\begin{itemize}
\item There are two possible worlds (two computation environments)
$E_T$ and $E_e$ that the computist can never distinguish between
them.

\item In Epistemic logic an agent \emph{knows}  a formula
$\varphi$ whenever the formula $\varphi$ is true in all possible
worlds which are not indistinguishable from the actual world for
the agent.

\item In the possible world $E_e$, $\mathrm{P\neq NP}$.

\end{itemize} thus,
\begin{center}
The computist who lives inside the Turing computation environment
$E_T$ does not know $\mathrm{P=NP}$.~\footnote{A formal version of
this argument is presented in the next paper ``Computation
Environment (2)".}
\end{center}

\section{Concluding Remarks}

$\mathrm{P}$~versus~$\mathrm{NP}$ is a major
 difficult problem in computational complexity theory, and it is shown that
lots of certain techniques  of complexity theory  are failed to
answer to this problem (see~\cite{kn:pnp2,kn:pnp1, kn:pnp3}).  We
constructed a computational environments $E_e$ which

\begin{itemize}
\item[1-] is physically plausible,

\item[2-] is enough powerful, and the computist who lives inside
the proposed environment can carry out  logical deductions,

\item[3-] similar to the human being who ignores his brain in
formalizing  notions of computation and time complexity, the
computist of the proposed environment  does not have access to the
inner structure and internal actions of the universal processor,

\item[4-] the universal processor of the computation persistently
evolve in a \emph{feasible} time, and

\item[5-] $\mathrm{P}$ is not equal to $\mathrm{NP}$ in the
proposed environment.

\end{itemize}
It seems that the above five items make it possible to know the
proposed computation environment as an alternative for the
computation environment which surrounded the human being.

This paper is a first manuscript of some serial articles which
would be appeared sequentially:
\begin{itemize}
\item[1-] \emph{A Persistent Evolutionary Semantics for Predicate
Logic.}  In (classic) model theory, the mathematician who studies
a structure is assumed as a god who lives out of the structure,
and the study of the mathematician does not effect the structure.
Against of this, we will propose a new semantics, called
persistently evolutionary semantics, for predicate logic that the
 meaning of functions and predicates symbols are not already predetermined,
 and  predicate and function symbols find their meaning through the
 interaction of the subject with the language. Then, we will prove formally that
the computist who lives inside  the Turing computation environment
can never \emph{know} $P=NP$.

\item[2-] We will discuss the human being as a computist of its
computation environment. Lots of people do not agree with us that
one is allowed to distinguish between syntax and semantics of
Turing machines.  Although, we proposed an interactive semantics
for computation and prove that $\mathrm{P}$ is not equal to
$\mathrm{NP}$ in this semantics, but one may note that the
computist who lives inside the computation environment $E_e$ does
not distinguish between the syntax and semantics of procedures in
$\Xi_e$ (since the computist ignores the inner structure of the
universal processor, and cannot be aware that whether he lives in
a static environment or a persistently evolutionary one).
Therefore, we will discuss that

\begin{itemize} \item it is possible that  the human being  lives in an interactive
environment (like the environment $E_e$),  \item the human being
does not distinguish between the syntax and semantics of Turing
machines (similar to the computist who lives inside $E_e$), and

\item  the human being, even by considering that the semantics and
syntax to be the same, cannot prove $\mathrm{P=NP}$.

\end{itemize}

\item[3-] There are lots of open problems in sciences. One of the
main reason that why $\mathrm{P~vs~NP}$ is so famous, is because
of the foundation of cryptography. We proved that $\mathrm{P\neq
NP}$ for interactive computation. Interactive computation is
physically plausible and could be considered as a new paradigm in
computer science. We will propose an interactive cryptosystem and
to show that it is safe, we reduce our interactive $\mathrm{P\neq
NP}$ statement to it.
\end{itemize}

 \vspace{.2in}

\noindent  \textbf{Acknowledgement}. I  would like to thank Amir
Daneshgar for his several careful reading of the manuscript, and
his very helpful comments,    advices, suggestions, and our
several useful discussions. I also would like to thank Mohammad
Saleh Zarepour for his comments, and Sharam Mohsenipour for his
objections. I would like to thank Farzad Didevar for his comments
and objections.  The idea of persistently evolutionary Turing
machines refers  back to the notion of dynamic computation in the
author's Phd thesis. The author would like to thank the supervisor
of his PhD thesis, Mohammad Ardeshir, for all of his warm
kindness, supports, and counsels.

\section{Appendix~A.}
In the appendix,  a very short review of hypercomputations and an
introduction of Brouwer's intuitionism is presented. To review the
Hypercomputation resources, the section~4~of~\cite{kn:ord} is
used, and the   review of Brouwer's choice sequences is derived
from~\cite{kn:ob}.

\subsection{Hypercomputation}\label{HYP}
In what follows, we briefly review a collection of
hypercomputational resources and those hypermachines that  use
these resources.

\vspace{0.5cm}

\noindent  \textbf{Non-recursive Information Sources}.

 The paradigm
hypercomputation starts with the \emph{o-machine}, proposed by
Turing in 1939~\cite{kn:om}. O-machine is a Turing machine
equipped with an oracle that is capable of answering questions
about the membership of a specific set of natural numbers.  If the
oracle set is recursive then the o-machine   gains no new power,
but if the oracle set is not itself computable by Turing machines,
the o-machine may compute an infinite number of non-recursive
functions. The o-machines use   non-recursive information
sources~\cite{kn:ord} as hypercomputational
resources~\cite{kn:ord}.

\emph{Coupled Turing machine}, introduced by Copeland and
Sylvan~\cite{kn:jcrs}, is a Turing machine with one or more input
channels, providing input to the machine while the computation is
in progress. The specific sequence of input determines the
functions that  the  coupled Turing machine can perform. It
exceeds a Turing machine if the sequence of input is
non-recursive. Like o-machines, coupled Turing machines use
\emph{non-recursive information sources}. Besides the above
discussed hypermachines,
  \emph{Asynchronous network of Turing machines}~\cite{kn:jcrs},
  \emph{Error prone Turing machines}~\cite{kn:ord},
and \emph{Probabilistic Turing machines}~\cite{kn:ptm} are
well-known hypermachines  that all use non-recursive information
sources as their hypercomputation resources.

\vspace{0.5cm}

 \noindent \textbf{Infinite Memory}.

  Another way to expand the capabilities of a
Turing machine is to allow it to begin with infinite number of
symbols initially inscribed on its tape. \emph{Turing machines
with initial inscription}   which have an explicitly infinite
amount of storage space are not physically plausible.

\vspace{0.5cm}

\noindent  \textbf{Infinite Specification.}

 An \emph{infinite state Turing
machine} is a Turing machine where the set of states is allowed to
be infinite. This type of machine has  an infinite amount of
transitions, with only a finite number of transition    from a
given state. This gives the Turing machine an infinite program of
which only a finite (but unbounded) amount of transitions  is used
in any given computation. The infinite state Turning machines
require \emph{infinite specification} which does not seems
physically plausible.

\vspace{0.5cm}

\noindent \textbf{Infinite Computation}.

 In the last century, Bertrand
Russell~\cite{kn:Russell}, Ralph Blake~\cite{kn:blake}, and
Hermann Weyl~\cite{kn:weyl} independently proposed the idea of a
process that performs its step in one unit of time and each
subsequent step in half the time of the step before~\cite{kn:ord}.
Therefore, such a process could complete an infinity of steps in
two time unit. The application of this temporal patterning to
Turing machines has been discussed briefly by Ian
Stewart~\cite{kn:is} and in much more depth by
Copeland~\cite{kn:bjc} under the name of accelerated Turing
machines. The hypercomputation resource is \emph{infinite
computation}. To achieve infinite computation through
acceleration, we rapidly run into conflict with physics.  For the
 tape head to get faster and faster, its speed converges to
 infinite.

Joel Hamkins and Andy Lewis presented another kind of
hypermachines that use infinite computation~\cite{kn:ittm}, named
\emph{infinite time Turing machines}. The infinite time Turing
machine is a natural extension of Turing machine to transfinite
ordinal times, the machine would be able to operate for
transfinite numbers of steps. An interesting case about
\emph{infinite time Turing machines} is that  it has   been proved
$P\neq NP$ for this model of computation~\cite{kn:ITM}.

\vspace{0.5cm}

\noindent \textbf{ Interaction}.

 Among other resources: non-recursive information
 sources, infinite memory, infinite specification, and infinite
 computation, the interaction seems to be    possible with our current
 physics.  Persistently evolutionary Turing machines which are
 introduced in this paper, use the interaction as a
 Hypercomputation resource.

A kind of hypermachines that use the  interaction as resource is
the class of persistent Turing machines
 (PTMs), \emph{multiple
machines with a persistent worktape preserved between
interactions}, independently introduced by Goldin and
Wegner~\cite{kn:weg} and Kosub~\cite{kn:kos}. \emph{Consistent}
PTMs, a subclass of PTMs, produce the same output string for a
given input string everywhere within a single interaction stream;
different interaction streams may have different outputs for the
same input. PTMs are a minimal extension of Turing machines that
express interactive behavior. The behavior of a PTM is
characterized by input-output streams instead of input-output
strings. \emph{Interaction streams} have the form
$(i_1,o_1),(i_2,o_2),...,$ where $i$'s are input strings and $o$'s
are output strings by PTM. For all $k$, $o_k$ is computed from
$i_k$ but preceded and can influence $i_{k+1}$. The set of all
interaction streams for a PTM $M$ consists its language, $L(M)$,
and two PTMs $M1$ and $M2$ are equivalent  if and only if
$L(M1)=L(M2)$.  Actually, PTMs extend computing to computable
nonfunctions over histories rather than noncomputable functions
over strings, whereas persistently evolutionary Turing machines
(introduced in this paper) extends computing to interactive
computable functions which are not \emph{predetermined. }

\subsection{Intuitionism}\label{intu}

Over the last hundred years, certain mathematicians have tried to
rebuilt mathematics on constructivist principles~\cite{kn:TD}.
However, there are considerable differences   between various
representatives of constructivist, and there exists no explicit
unique answer to what a \emph{constructive object}  or a
\emph{constructive method} is.  In contrast to other
constructivist schools as Bishop's and Markov's~\cite{kn:db}, for
Brouwer, the notion of ``constructive object'' is not restricted
to have a numerical meaning, or to   be presented as `words' in
some finite alphabet of symbols~\cite{kn:TD}. According to
Brouwer, mathematical objects are mental constructions, a {\em
languageless} activity and independent of logic. Brouwer
recognized   the {\em choice sequences} as legitimate mathematical
objects~\cite{kn:TD}.

 Imagine you have a
collection of objects at your disposal, let's say the natural
numbers. Pick out one of them, and note the result. Put it back
into the collection, and choose again. Since you have the
\emph{ability to choose freely}, you may choose a different one,
or the same again. Record the result, and put it back. You may
make further choices and keeping on. A \emph{choice sequence} is
what you get if you think of the sequence you are making as
potentially infinite~\cite{kn:ob}. Initial segments are always
finite. We cannot make an actually infinite number of choices, but
we can always extend an initial segment by making a further
choice. The following characteristic of the choice sequences is
crucial in our consideration:
\begin{quote}
{\footnotesize The subject successively chooses objects,
restriction on future choices, restriction on restriction of
future choices, etc.~(\cite{kn:vv}~page 6)}
\end{quote}
The object of classical mathematics have their properties
independently from us and are \emph{static}. Choice sequences, in
contrast, depends on the subject (who has to make the choice), and
they change through time. They are individual \emph{dynamic}
objects that come into being, at the moment that the subject
decides to intend them, and with each further choice, they grow,
and they are not necessarily \emph{predetermined} by some law.
\begin{quote}
{\footnotesize{\bf static-dynamic} An object is static exactly if
at no moments parts are added to it, or removed from it. It is
dynamic exactly if at some moment parts are added to it, or
removed from it.~(\cite{kn:va},~page 12)}
\end{quote}

Since choice sequences are dynamic objects and are accepted as
 intuitionistic mathematical objects, the notion of
``constructive method'' cannot be captured by Turing
computability. In addition, in intuitionism, the notion of
decidability differs from the notion of recursiveness. Although
any recursive set is decidable, the converse is not true. From
intuitionistic view, a subset $A$ of $\mathbb{N}$, is
\emph{decidable} if and only if there exists a sequence $\alpha\in
2^{\mathbb{N}}$, such that, for every $n$, $n\in A$ if and only if
$\alpha(n)=1$~\cite{kn:deci}. It is not required that the
 sequence $\alpha$ is given by a finite algorithm, it can be a
 choice sequence.


\begin{thebibliography}{10}
\bibitem{kn:pnp1} S. Aaronson, {\em Is $P$ versus $NP$ formally independent},
ACM Transaction on Computing Theory, 2009.

\bibitem{kn:pnp3} S. Aaronson, A. Wigderson {\em Algebrization: A New Barrier in Complexity Theory},
Bulletin of the European Association for Theoretical Computer
Sciences, 2003.



\bibitem{kn:deci} M. Ardeshir, R. Ramezanian, {\em Decidability and Specker
sequences in intuitionistic mathematics}, Journal of Mathematical
Logic Quarterly, 2009.

\bibitem{kn:parad} M. Ardeshir, R. Ramezanian, {\em A Solution to the Surprise Exam Paradox\\
in Constructive Mathematics}, To appear in journal of Review of
Symbolic Logic.

\bibitem{kn:arora} S. Arora, B. Barak,
 {\bf Computational Complexity, a modern approach}. Cambridge
 University Press, 2007.

\bibitem{kn:va} M. van Atten,
 {\bf Phenomenology of choice sequences}. Ph.D thesis, Utrecht
 University, 1999.

\bibitem{kn:BH} M.~van~Atten,
{\bf Brouwer Meets Husserl, on the phenmenology of choice
sequences}, Springer, 2007.

\bibitem{kn:vv} M. van Atten, D. van Dalen.
{\em Arguments for the continuity principle}, Journal of Symbolic
Logic, Volume 8, Number 3, Sept. 2002

\bibitem{kn:ob} M.~van~Atten,
{\bf On Brouwer}, Wadsworth Philosophers Series, 2004.


\bibitem{kn:blake} R.M. Blake.
{\em The Paradox of Temporal Process},   Journal of Philosophy,
23:645-654, 1926.


 \bibitem{kn:db} D.Bridges, F.Richman.
 {\bf Varieties of Constructive Mathematics}. London Mathematical
 Society lecture note series, 1986.

\bibitem{kn:bjc} B. Jack Copeland,
{\em Even Turing machines can Compute Uncomputable Functions},
page 150-164,   1998.

\bibitem{kn:jcrs} B. Jack Copeland and R. Sylvan,
{\em Beyond the Universal Turing Machine}, Australian Journal of
Philosophy, 77:46-66, 1999.

\bibitem{kn:dit3} H. van Ditmarsch, W. van der Hoek, and B. Kooi,
{\bf Dynamic Epistemic Logic}, Springer, 2008.

\bibitem{kn:CFLA} R. L. Epstein, W. A. Carnielli,
{\bf Computability, computable functions, logic, and the
foundations of mathematics}, ARF, 2008.



\bibitem{kn:weg07} D. Goldin and P. Wegner.
{\em The Interactive Nature of Computing: Refuting the Strong
Church-Turing Thesis}, Minds and Machines 18(1; Spring): 17-38,
2008.

\bibitem{kn:intcom} D. Goldin, S. A. Smolka, P. Wegner (Ed.), {\bf
Interactive Computation, the new paradigm}, Springer, 2006.


\bibitem{kn:weg} D. Goldin and P. Wegner.
{\em Presistnet as a form of interaction}, Technical Report
CS-98-07, Brown University, Department of Computer Science, 1998.

\bibitem{kn:PTM} D. Q. Goldin, {\em Persistent Turing Machines as a Model of Interactive Computation},
In Foundation of information and kKnowledge systems, K. D. Schewe,
and B. Thalheim, eds., vol 1762 of Lecture Notes in Computer
Scince. Springer-Verlag, Berlin, 2000, pp. 116-135.

\bibitem{kn:ittm} J. D. Hamkins and A. Lewis.
{\em Infinite Time Turing Machines}, Journal of Symbolic Logic,
65: 567-604, 1998.

\bibitem{kn:kos} Sven Kosub.
{\em Persistent Computations}, Theoretische Informatik Tech.
Report, U.~Wurzburg, 1998.

\bibitem{kn:ptm} K.D. Leeuw, E.F. Morre, C. E. Shannon, and N. Shapiro.
{\em Computability by probabilistic machines}, 183-212 in C.E.
Shannon and J. NcCarthy, editors, Automata Studies, Princeton
University Press, Princeton, N.J., 1956.

\bibitem{kn:model} D. Marker, {\bf Model Theory: an introudction},
Springer, 2002.

\bibitem{kn:TO} T. Ord, {\em The many forms of hyeprcomputation},
Journal of Applied Mathematics and Computation 178, 143--153,
2006.

\bibitem{kn:ord} T. Ord.
{\em  Hypercomputation: Computing more than the Turing machine},
Technical Report arXiv:math.LO/0209332, University of Melbourne,
Melbourne, Australia, September 2002.

\bibitem{kn:phdr} R. Ramezanian,
{\bf The Temporal Continuum  and Dynamic Computation}, PhD thesis,
August 2008.

\bibitem{kn:pnp2} A. Razborov, S. Rudich, {\em Natural Proofs},
Journal of Computer and System Sciences, 1996.

\bibitem{kn:Russell} B.A.W. Russell.
{\em  The Limits of Empiricism}, Proceedings of the Aristotelian
Society, 36: 131-150, 1936.

\bibitem{kn:ITM} R. D. Schindler, {\em $P\not= NP$ for Infinite Time Turing Machines},
Monatshefte fur Mathematik 139, 4. 335-340, 2003.

\bibitem{kn:is} Ian Stewart.
{\em  Deciding the Undecidable}, Nature, 352:664-665, 1991.

\bibitem{kn:AS} A.
Syropoulos, {\bf Hypercomputation: computing beyond the
Church-Turing barrier}, Springer, 2008.

\bibitem{kn:TD} A. S. Troelstra, D. van Dalen,
{\bf Constructivism in Mathematics, An introduction}, Vol. 1,
North-Hollan d, 1988.

\bibitem{kn:om} A.M. Turing.
{\em  Systems of Logic Based on the Ordinals}, Proceeding of the
London Mathematical Society, 45:161-228, 1939.


\bibitem{kn:weyl} H. Weyl. {\em Philosophie der Mathematik und
Naturwissenschaft},
  R. Oldenburg, Munich, 1927.

\end{thebibliography}
\end{document}